\pdfoutput = 1

\documentclass[11pt]{article}
\usepackage{array,comment}

\newcommand{\Var}{\rm Var}
\usepackage{graphicx}
\usepackage[title]{appendix}

\usepackage{amsmath,amssymb,amsthm}
\newtheorem{lemma}{Lemma}
\usepackage{enumerate}

\usepackage[normalem]{ulem}

\newcommand{\eqrevision}[2]{
\text{$\displaystyle #1$} 
}

\begin{document}




\title{Understanding the role of phenotypic switching \\ in cancer drug resistance 
}


\author{Einar Bjarki Gunnarsson$^{1,\ast}$ \and Subhajyoti De$^2$ \and Kevin Leder$^1$ \and Jasmine Foo$^{3,\ast}$}

\date{%
    \footnotesize $^1$Department of Industrial and Systems Engineering, University of Minnesota, Twin Cities, MN 55455, USA. \\[2pt]%
    $^2$Center for Systems and Computational Biology, Rutgers Cancer Institute of New Jersey, NJ 08903, USA. \\[2pt]%
    $^3$School of Mathematics, University of Minnesota, Twin Cities, MN 55455, USA. \\[-2pt]
    $^\ast$Corresponding authors: jyfoo@umn.edu (Jasmine Foo), gunna042@umn.edu (Einar Bjarki Gunnarsson). \vspace*{-6pt}
}

\maketitle

\begin{abstract}
The emergence of acquired drug resistance in cancer represents a major barrier to treatment success.
While research has traditionally focused on genetic sources of resistance, recent findings suggest that cancer cells can acquire transient resistant phenotypes via epigenetic modifications and other non-genetic mechanisms. Although these resistant phenotypes are eventually relinquished by individual cells, they can temporarily 'save' the tumor from extinction and enable the emergence of
more permanent resistance mechanisms.
These observations have generated interest in the potential of epigenetic therapies
for long-term tumor control or eradication. In this work, we develop a mathematical model to study how phenotypic switching at the single-cell level
affects resistance evolution in cancer.
We highlight unique features of non-genetic resistance,
probe the evolutionary consequences of epigenetic drugs and explore potential therapeutic strategies.
We find that even short-term epigenetic modifications and stochastic fluctuations in gene expression can drive long-term drug resistance in the absence of any bona fide resistance mechanisms.
We also find that an epigenetic drug that slightly perturbs the average retention of the resistant phenotype can turn guaranteed treatment failure into guaranteed success.
Lastly, we find that combining an epigenetic drug with an anti-cancer agent can significantly outperform monotherapy, and that treatment outcome is heavily affected by drug sequencing. \\[-3pt]

\noindent {\em Keywords:} Mathematical modeling, cancer drug resistance, evolutionary dynamics, phenotypic switching, epigenetics.  \\[-3pt]

\noindent {\em Licensing:} This article is licensed under a Creative Commons Attribution, NonCommercial, NoDerivatives license (CC BY-NC-ND). 
\end{abstract}



\section{Introduction}
While cancer has traditionally been considered a genetic disease driven by Darwinian evolution at the somatic level,
it is now increasingly recognized that non-genetic sources of phenotypic variation may play an important role in tumor initiation, tumor progression and the evolution of drug resistance \cite{Brock2009, jones2007epigenomics, merlo2006cancer, marusyk2012intra, flavahan2017epigenetic, brown2014poised, wilting2012epigenetic,easwaran2014cancer}.
Common sources of non-genetic heterogeneity
include DNA methylation, histone modifications and other epigenetic mechanisms that alter gene expression, without changing the genetic code,
by controlling DNA accessibility during transcription, replication and repair \cite{woodcock2006chromatin}.
Since these mechanisms 
frequently operate at a significantly
faster rate
than genetic mutations, they can serve as a substrate for natural selection and permanently influence tumor evolution in the absence of any genetic events \cite{Brock2009,jones2007epigenomics,brown2014poised}.
Another common source of variation
in gene expression is
the
inherent stochasticity of intracellular biochemical reactions, which includes transcriptional noise.
This stochasticity may give rise to heritable expression levels,
albeit on the short time scale of one to a few cell divisions, whereas retention of epigenetic modifications has been estimated on the order of
$10$-$10^5$ cell divisions
\cite{niepel2009non, sigal2006variability, cohen2008dynamic}. As we will find, even such short-term phenotypic states can
dramatically impact
the course of
tumor evolution.

Here, we are primarily interested in the role of non-genetic mechanisms in conferring acquired resistance to anti-cancer treatment, as has been explored in several recent experimental works. In Sharma et al.~\cite{sharma2010chromatin}, for example, the authors investigate the acute response of several cancer cell lines\footnote{Non-small cell lung cancer, melanoma, colorectal cancer, breast cancer and gastric cancer.} to targeted anti-cancer agents, and they consistently observe the emergence of a drug-tolerant phenotype (DTP) that is `transiently acquired and relinquished by individual cells within the population at a low frequency'. The authors draw an analogy between DTP's and slowly-proliferating antibiotic-tolerant `persisters' commonly observed in microbial populations \cite{sharma2010chromatin,balaban2004bacterial, dhar2007microbial}, whose survival within a more rapidly proliferating population represents an evolutionary means of 'bet-hedging' against potential environmental stresses \cite{gaal2010exact, muller2013bet, nichol2016stochasticity, jolly2018phenotypic}.
Liau et al.~\cite{liau2017adaptive} and Roesch et al.~\cite{roesch2010temporarily,roesch2013overcoming} similarly describe slow-cycling DTP's in glioblastoma and melanoma, respectively,
and Knoechel et al.~\cite{knoechel2014epigenetic} identify a reversible drug-tolerant state in leukemia which appears to be epigenetically mediated. For even further examples of experimental studies describing (often stem-like) non-genetic phenotypces associated with tumorigenic potential or drug resistance in cancer, we refer to e.g.~\cite{chang2008transcriptome, quintana2010phenotypic, chaffer2011normal, seghers2012successful,sun2014reversible}, as well as recent reviews by Reyes and Lahav \cite{reyes2018leveraging} and Salgia and Kulkarni \cite{salgia2018genetic}. We now turn our attention to such studies that also incorporate a modeling component.

In Gupta et al.~\cite{gupta2011stochastic}, the authors employ a mathematical model of stochastic switching between three cell types to infer the rates at which breast cancer cells transition between a selectively resistant stem-like state and two non-stem-like states. Su et al.~\cite{su2017single} show that phenotypic switching between a drug-sensitive and drug-resistant state in melanoma is well-captured by a similar model, and their study reveals the critical role played by drug-induced adoption of the resistant state, relative to selection of preexisting cells in this state. Goldman et al.~\cite{goldman2015temporally} further provide evidence of chemotherapy-induced switching to a resistant ${\rm CD44}^{\rm Hi}{\rm CD24}^{\rm Hi}$ expression status in breast cancer, and Pisco et al.~\cite{pisco2013non} show that vincristine resistance in leukemia, mediated by overexpression of the multidrug resistance protein 1 (MDR1), is primarily due to
therapy-accelerated adoption of the overexpressed state.
Thus, while reversible drug-tolerant cells may arise naturally in drug-na\"ive cell populations, their emergence can also be directly influenced by anti-cancer treatment. Further complicating the picture, Craig et al.~\cite{craig2019cooperative} have hypothesized that genotypically and phenotypically distinct cells can cooperate to induce the adoption of drug tolerance, although the potential mechanism behind such cooperation remains unclear.

Whereas transiently resistant cells serve the immediate function of protecting the tumor population from extinction, their ultimate role appears to be to set the stage for the evolution of more permanent resistance mechanisms, both of the genetic and epigenetic kind. In Sharma et al.~\cite{sharma2010chromatin}, for example, the authors report that during prolonged exposure of {\em EGFR}-mutant non-small cell lung cancer (PC9) to erlotinib, 
a fraction of drug-tolerant cells become capable of proliferating normally in drug, and that this more aggressive phenotype reverts less readily to sensitivity once removed from drug.
During even more prolonged exposure, Ramirez et al.~\cite{ramirez2016diverse} find that PC9 persister cells give rise to diverse genetic resistance mechanisms, including {\em de novo} adoption of the ${EGFR}^{\rm T790M}$ gatekeeper mutation commonly observed in the clinic (see also Hata et al.~\cite{Engelman2016}). Shaffer et al.~\cite{shaffer2017rare} report findings conceptually similar to Sharma et al.~for melanoma cells treated with vemurafenib,
and they further report that prolonged drug exposure induces epigenetic reprogramming of the drug-tolerant state, ultimately leading to a stable drug-resistant phenotype. 

The above studies indicate that tumor cells in a wide variety of cancer types 
have the ability to adopt reversible drug-resistant phenotypes, which can in turn facilitate
the eventual acquisition of bona fide resistance mechanisms. Such phenotypes can both preexist anti-cancer treatment and be specifically induced or accelerated by therapy.
The relatively fast rate at which non-genetic phenotypes can be adopted, as compared to resistance-conferring mutations, poses a major challenge for clinical strategies.
A deeper understanding of these dynamics is crucial to furthering our understanding of cancer and to informing novel therapeutic efforts.
Here, we develop a mathematical model
to investigate the evolutionary dynamics of a cell population that is able to employ transiently resistant states as a survival strategy.
Our goal is to gain quantitative insights into resistance evolution in this setting, to highlight some of its unique characteristics, and to discuss how an understanding of these dynamics can inform the design of treatment strategies.

\begin{figure}
\centering
\includegraphics[scale=1]{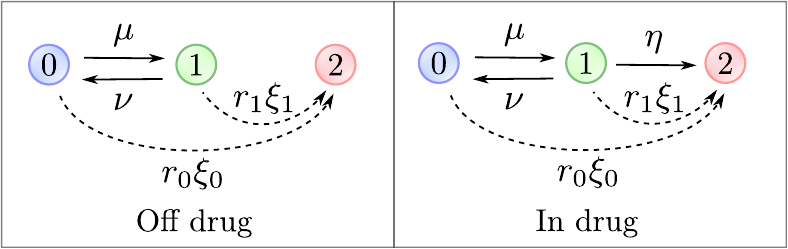} 
\caption{
Graphical representation of the model.
Type-0 (drug-sensitive) cells transition to type-1 (transiently resistant) cells at rate $\mu$ and type-1 cells transition back at rate $\nu$. Each phenotype has distinct growth characteristics, with $r_0$ and $r_1$ denoting the rates of cell division and $d_0$ and $d_1$ denoting the rates of cell death. In the presence of an anti-cancer agent, the transiently resistant type-1 phenotype undergoes epigenetic reprogramming to stable drug resistance at rate $\eta$. Additionally, each cell type can acquire a resistance-conferring mutation at rate $\xi_0$ and $\xi_1$ per cell division, respectively, on and off drug.
Stably resistant (type-2) cells divide at rate $r_2$ and die at rate $d_2$.
}
\label{twotypemodel1} \vspace*{-6pt}
\end{figure}

\vspace*{-6pt}
\section{Model description}

We consider a multi-type branching process model \cite{athreya2004branching}, in which cells switch stochastically between two distinct phenotypes, the drug-sensitive type-0 phenotype and the drug-resistant type-1 phenotype. A type-0 cell divides into two type-0 cells at (Poisson) rate $r_0$, it dies at rate $d_0$ and adopts the type-1 phenotype at rate $\mu>0$, with $\lambda_0:=r_0-d_0$ the net birth rate.\footnote{This means that each type-0 cell waits an exponential amount of time with rate $a_0 := r_0+d_0+\mu$ (i.e.~the waiting time is exponentially distributed with mean $1/a_0$) before it either divides with probability $r_0/a_0$, dies with probability $d_0/a_0$, or adopts the type-1 phenotype with probability $\mu/a_0$.}
A type-1 cell divides at rate $r_1$, dies at rate $d_1$ and reverts to type-0 at rate $\nu>0$, with $\lambda_1:=r_1-d_1$ the net birth rate (Fig \ref{twotypemodel1}).
To capture the drug-sensitivity of type-0 cells and drug-resistance of type-1 cells, we assume that $\lambda_0<0$ and $\lambda_1>0$ in the presence of an anti-cancer agent.
Although this model assumes that switching between phenotypes can occur at any time during the cell cycle, it can easily be adjusted to allow switching to only occur at cell divisions (Appendix \ref{app:alternateversion}).

These general two-type switching dynamics enable description of a variety of sources of phenotypic heterogeneity;
for example, short-term drug-tolerant states conferred by transcriptional noise can be captured by a large reversion rate $\nu$, while longer-term states induced by epigenetic phenomena are captured by a smaller $\nu$.
We also note that the switching rates $\mu$ and $\nu$ are in general distinct, since the mechanism underlying phenotypic switching is in general asymmetric.
A simple example is DNA methylation/demethylation, where {\em de novo} methylation is carried out by the DNA methyltransferases DNMT3a and DNMT3b, while demethylation usually occurs due to a failure of the {\em maintenance} methyltransferase DNMT1 to faithfully preserve methylation patterns during cell division \cite{esteller2008epigenetics}.
In what follows, we will usually refer to the transition from type-0 to type-1 as an {\em epimutation} and the transition back as a {\em reversion}, while keeping more general sources of non-genetic variation (e.g.~transcriptional noise) in mind.

To capture the evolution of more permanent resistance mechanisms, we assume that under anti-cancer treatment, the type-1 phenotype undergoes epigenetic reprogramming to a stably resistant phenotype at rate $\eta$. Alternatively,
type-0 and type-1 cells can acquire a resistance-conferring mutation
at rate $\xi_0$ and $\xi_1$ per cell division, respectively (i.e.~the mutation rate per time unit is $r_0\xi_0$ for type-0 cells and $r_1\xi_1$ for type-1 cells),
both in the presence and absence of the anti-cancer agent (Fig \ref{twotypemodel1}).
Stably resistant (type-2) cells
divide at rate $r_2$ and die at rate $d_2$, with $\lambda_2 := r_2-d_2 > 0$ the net birth rate.
We assume throughout that no type-2 cell is present at detection, focusing on how acquired resistance develops during anti-cancer treatment.

We note that the model outlined above assumes that epigenetic reprogramming from type-1 to type-2 occurs in a single stochastic event.
Single-stage reprogramming is consistent e.g.~with a model of epigenetic gene silencing under recruitment of chromatin regulators suggested in a recent paper by Bintu et al.~\cite{bintu2016dynamics}, and with Brown et al.'s \cite{brown2014poised} notion of 'epigenetically poised' states that evolve to fixed acquired-resistant states via DNA methylation. Conversely, a multi-stage (or more continuous) model appears more consistent e.g.~with the findings of Sharma et al.~\cite{sharma2010chromatin} and Shaffer et al.~\cite{shaffer2017rare} described above, although the data in these works are insufficient to infer an exact model. The distinction between single-stage and multi-stage reprogramming will not be important to much of our investigation, since in Sections \ref{sec:solelyswitching} and \ref{sec:tumorsurvival}, we focus on behavior in the absence of permanent resistance mechanisms ($\eta=0$ and $\xi_0=\xi_1=0$). In Sections \ref{sec:pathway} and \ref{sec:combination}, we work with the minimal single-stage model, as this allows us to gain theoretical insights into the role played by transiently resistant phenotypes in facilitating the evolution of more permanent resistance mechanisms in the simplest possible setting, in terms of a single reprogramming parameter $\eta$. 
Our results in these sections will still be meaningful for the multi-stage case if we interpret the type-2 phenotype more generally as a stabler and more aggressive form of epigenetic resistance, not necessarily representing a permanent acquired-resistant state.
Moreover, once multi-stage models are better understood mechanistically and quantitatively,
our analysis can be easily extended to capture these more complex dynamics, using the same multi-type branching process framework as employed below.

\subsection{Parametrization}
For demonstration purposes, we adopt a baseline parameter regime chosen to mimic {\em in vitro} behavior of {\em EGFR}-mutant non-small cell lung cancer (PC9) reported by Sharma et al.~\cite{sharma2010chromatin} and Hata et al.~\cite{Engelman2016}. We estimate the birth and death rates of all phenotypes and the epimutation rate $\mu$ using these works, but refer to other works for estimation of the reversion rate $\nu$ and the rates $\eta$, $\xi_0$ and $\xi_1$ of stable resistance acquisition. Details are provided in Appendix \ref{app:parametrization}. The exact parameter values are not important to our investigation, but rather the qualitative setting encoded in the regime: Type-0 cells proliferate rapidly in the absence of an anti-cancer agent and die rapidly in its presence, while type-1 cells are able to maintain slow net proliferation under the anti-cancer agent. Epigenetic reprogramming is further assumed to occur at a faster rate than resistance-conferring mutations.
To ensure that our main insights are not particular to our chosen regime, but that they apply more generally across cancer cell populations capable of adopting slow-cycling, transiently resistant phenotypes, we will usually examine a range of possible switching dynamics.  We furthermore discuss the sensitivity of our results to main model parameters, and perform robustness analysis where appropriate.
We finally note that the rapid {\em in vitro} dynamics that underlie Figures \ref{fig:criticalregion_expbehavior} to \ref{fig:interval} can be translated into slower {\em in vivo} dynamics through appropriate rescaling of time, as is discussed in Section \ref{sec:results}.

\section{Mean behavior and survival probability} \label{sec:meanbehavior}

We begin by deriving expressions for the average number of type-0 and type-1 cells alive at any time $t$, assuming the absence of permanent resistance mechanisms (i.e.~$\eta=0$ and $\xi_0=\xi_1=0$).
These expressions will be useful both for analyzing long-term tumor evolution and for estimating the time at which permanent resistance first arises (Section \ref{sec:combination}).

Assume that the tumor consists of $n$ type-0 cells and $m$ type-1 cells at the start of anti-cancer treatment. 
We will both be interested in the case $m=0$, where resistance is mediated by drug-induced adoption of the type-1 phenotype, and $m \gg 0$, where transiently resistant cells preexist treatment (in significant numbers).
The switching dynamics of tumor cells are encoded in the
{\em infinitesimal generator} for the process, \vspace*{-12pt}

\begin{align} \label{eq:infgen}
{\bf A} = \begin{bmatrix} \lambda_0-\mu & \mu \\ \nu & \lambda_1-\nu \end{bmatrix},
\end{align}
where $\lambda_0-\mu$ (resp.~$\lambda_1-\nu$) is the net rate at which a type-0 (type-1) cell produces another type-0 cell, and $\mu$ (resp.~$\nu$) is the transition rate from type-0 to type-1 (type-1 to type-0). 
If we let $\phi_0(t)$ (resp.~$\phi_1(t)$) denote the mean number of type-0 (type-1) cells alive at time $t$, we can calculate these means as \vspace*{-12pt}

    \begin{align} \label{eq:meannumtypes}
[\phi_0(t) \;\; \phi_1(t)]= [n \,\; m] \; \exp({\bf A}t),
\end{align}
which allows use to derive the following expressions:  \vspace*{-12pt}

\begin{align}  \label{eq:meanindividual}
\begin{split}
\phi_0(t)  &= \frac{n\delta + m}{\delta + \gamma} e^{\sigma t} + \frac{n\gamma -m}{\delta+\gamma} e^{\rho t},  \\
\phi_1(t) &= \frac{\gamma(n\delta+m)}{\delta+\gamma} e^{\sigma t} - \frac{\delta(n\gamma-m)}{\delta+\gamma} e^{\rho t}.
\end{split}
\end{align}
Details of the derivation are provided in Appendix \ref{app:generalbranching}.
The rate constants $\sigma$ and $\rho$, with $\sigma>\rho$, are the (real) eigenvalues of the infinitesimal generator ${\bf A}$, given by \begin{align} \label{eq:eigenvaluesmain}
\frac{(\lambda_0 - \mu) + (\lambda_1 -\nu)\pm \sqrt{((\lambda_0 - \mu)-(\lambda_1 -\nu))^2+4\mu\nu}}{2}.
\end{align}
In addition,\vspace*{-12pt}

\begin{align} \label{eq:gamma}
    \gamma := (\sigma-(\lambda_0-\mu))/\nu > 0
\end{align}
is the long-run ratio between type-1 and type-0 cells in the tumor population, and \vspace*{-12pt}

\begin{align}
    \delta := ((\lambda_0-\mu)-\rho)/\nu>0
\end{align}
is the long-run ratio between the size of a clone derived from a single type-0 vs.~a single type-1 cell.
Note that the mean behavior of the process can either be expressed as a function of the fundamental constants $\lambda_0$, $\lambda_1$, $\mu$ and $\nu$, which capture single-cell level dynamics, or as a function of the derived quantities $\gamma$, $\delta$, $\sigma$ and $\rho$, which capture long-run population-level behavior and may be more easily observable in an experimental setting.
Also note that the above expressions only depend on the birth and death rates of each cell type through their net birth rates $\lambda_0=r_0-d_0$ and $\lambda_1=r_1-d_1$.

As expression (\ref{eq:meanindividual}) indicates, the long-run behavior of the tumor population is determined by the sign of the rate constant $\sigma$.
The population survives with positive probability
if and only if $\sigma>0$, in which case it is said to be {\em supercritical}, while extinction is guaranteed for $\sigma < 0$, in which case it is {\em subcritical} (see e.g.~\cite{athreya2004branching} for further information). For a supercritical population, the survival probability can be computed by solving a system of two nonlinear equations, as is outlined in Appendix \ref{app:extinct}.

\section{Results} \label{sec:results}

\subsection{
Resistance driven solely by phenotypic switching} \label{sec:solelyswitching}
We begin by
investigating 
whether
phenotypic switching can drive long-term resistance to continuous
anti-cancer treatment, even in the absence of permanent resistance mechanisms (i.e.~$\eta = 0$ and $\xi_0=\xi_1=0$). 
By examining when $\sigma>0$, we can show that 
tumor survival is possible 
(i.e.~the tumor cell population is supercritical)
if and only if \vspace*{-12pt}

\begin{align} \label{eq:supercriticality} 
\nu \lambda_0 + \mu \lambda_1 > \lambda_0 \lambda_1 
\end{align}
(see Appendix \ref{app:lemma1} for details). By rewriting this condition as (recall that $\lambda_0<0$ under anti-cancer therapy) 
\begin{align} \label{eq:supercriticality1}
\nu/\lambda_1-\mu/|\lambda_0|<1,
\end{align}
we see that the rates $\nu/\lambda_1$ and $\mu/|\lambda_0|$ of phenotypic switching, relative to net growth (or net decay) of each phenotype,
determine whether the tumor can persist under therapy.
It is furthermore easy to see that a sufficient condition for (\ref{eq:supercriticality1}) is 
\begin{align} \label{eq:suffcond}
\nu/\lambda_1 \leq 1 \quad\text{i.e.}\quad 1/\nu \geq 1/\lambda_1,    
\end{align}
which can be interpreted as a simple condition on the time scale of 'phenotypic memory':
The tumor population has a chance of surviving treatment whenever the average memory $1/\nu$ of the resistant state equals or exceeds $1/\lambda_1$.

If we assume that the growth characteristics of type-0 and type-1 cells ($\lambda_0$ and $\lambda_1$) are fixed, the condition \eqref{eq:supercriticality1} for supercriticality can be viewed as describing the region in the $(\mu,\nu)$ plane that lies below the 'critical line' \vspace*{-12pt}

\begin{align} \label{eq:criticalitycurve}
    \nu = \lambda_1/|\lambda_0| \cdot \mu + \lambda_1.
\end{align}
Since $\mu$ and $\nu$ are small, and they may differ by orders of magnitude, it is more instructive to visualize their relationship on a logarithmic scale. In Figure \ref{fig:criticalregion_expbehavior}a, we show this $\log$-scale 'critical curve' for two cases: (i) $\lambda_1 \ll |\lambda_0|$,  i.e.~type-0 cells die rapidly and type-1 cells proliferate slowly under anti-cancer treatment, as in our baseline parameter regime, and (ii) $\lambda_1 \sim |\lambda_0|$, i.e.~the type-1 net birth rate is of the same order as the type-0 net death rate. In both cases, the tumor population is supercritical for $\nu \leq \lambda_1$ (region A), which is our sufficient condition from \eqref{eq:suffcond}, but the population can still be supercritical for $\nu > \lambda_1$ (region B for $\lambda_1 \ll |\lambda_0|$ and regions B+C for $\lambda_1 \sim |\lambda_0|$), the degree to which is controlled by the slope $\lambda_1/|\lambda_0|$ in \eqref{eq:criticalitycurve}.

\begin{figure}[t!]
\centering
\hspace*{-1cm}\includegraphics[scale=1]{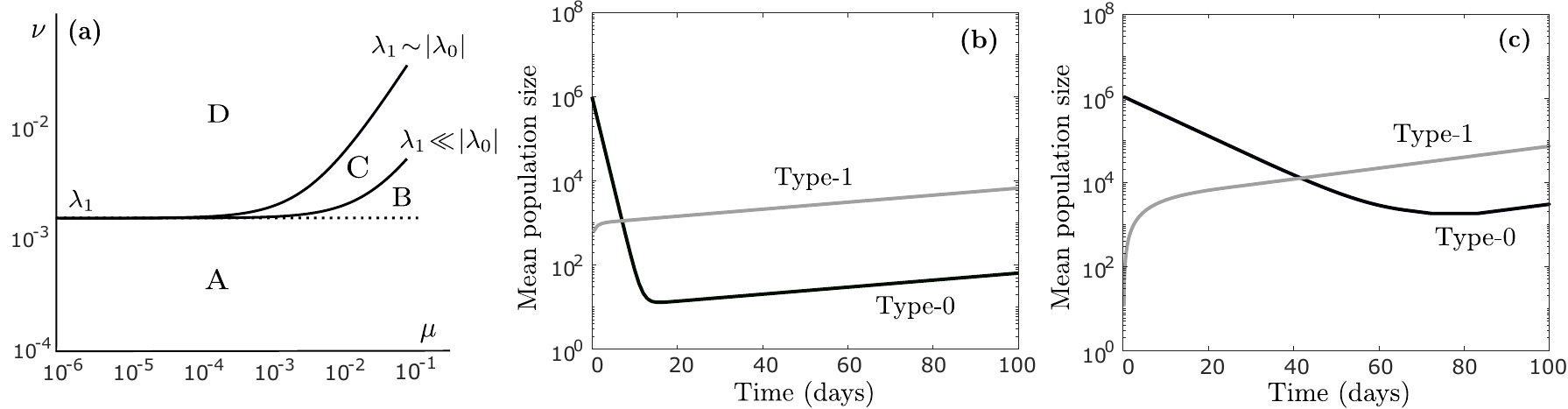} 
\caption{
{\bf (a)} Graphical depiction of the region in the $(\mu,\nu)$ plane where the tumor cell population is supercritical (the region below the line in \eqref{eq:criticalitycurve}), displayed here on a logarithmic scale for two cases ($\lambda_1 \ll |\lambda_0|$ and $\lambda_1 \sim |\lambda_0|$). The two curves are drawn assuming $r_1 = 0.0162$ (per hour) and $d_1=0.015$, and $r_0=0.04$ and $d_0=0.08$ for the $\lambda_1 \ll |\lambda_0|$ case (baseline parameter regime), and $r_0 = 0.004$ and $d_0 = 0.008$ for the $\lambda_1 \sim |\lambda_0|$ case.
Since the slope of \eqref{eq:criticalitycurve} is $\lambda_1/|\lambda_0|>0$, and it intercepts the $\nu$-axis at $\nu = \lambda_1$, the population is supercritical whenever $\nu \leq \lambda_1$, independently of the relationship between $\lambda_0$ and $\lambda_1$ (region A). When phenotypic memory is short ($\nu>\lambda_1$), the population can still be supercritical, the degree to which depends on the slope $\lambda_1/|\lambda_0|$ (region B for $\lambda_1 \ll |\lambda_0|$ and regions B+C for $\lambda_1 \sim |\lambda_0|$).
{\bf (b)} Time-evolution of the expected number (log-scale) of type-0 cells (dark curve) and type-1 cells (light curve) during continuous anti-cancer treatment, assuming no transiently resistant cell is present at detection ($m=0$), calculated using (\ref{eq:meanindividual}).
Parameter values are $r_0 = 0.04$ (per hour), $d_0 = 0.08$, $r_1 = 0.0162$, $d_1 = 0.015$, $\mu = 4 \cdot 10^{-5}$, $\nu = 4 \cdot 10^{-4}$, $\eta = 0$, $\xi_0=\xi_1=0$, $n = 10^6$ and $m=0$. 
{\bf (c)} Same as ({\bf b}), except now $r_0 = 0.004$ (per hour) and $d_0 = 0.008$.
}
\label{fig:criticalregion_expbehavior}
\end{figure}

We conclude 
from the above
that phenotypic switching can in fact drive long-term drug resistance in the absence of more permanent resistance mechanisms,
and we identify two qualitatively distinct
evolutionary pathways to such resistance:
\begin{enumerate}
\item $\nu \leq \lambda_1$: Population survival is driven by 
sufficiently long retention of the resistant phenotype, independently of the rate of epimutation ($\mu$) and type-0 sensitivity to the anti-cancer agent ($\lambda_0$).
\item $\nu > \lambda_1$ and $\mu > \lambda_0(1-\nu/\lambda_1)$: Type-1 cells lose the resistant phenotype too quickly to sustain the tumor by themselves, but this loss is compensated by sufficiently fast adoption of the resistant state.
\end{enumerate}
If we assume $d_1=0$, i.e.~type-1 cells have a stem-like ability to proliferate indefinitely,
the first condition implies that even single-generation phenotypic memory may be sufficient to confer long-term resistance, while the second condition implies that even non-heritable traits, 
e.g.~stochastic variation in gene transcription,
may be able to save the tumor from extinction.

To further elucidate the coevolutionary dynamics between type-0 and type-1 cells, we show in Figure \ref{fig:criticalregion_expbehavior}b the long-term expected behavior of the population in the baseline parameter regime, assuming no transiently resistant cell is present at detection ($m=0$).
When the anti-cancer agent is applied, sensitive type-0 cells initially die at a fast rate, while a small fraction of them adopts the resistant type-1 phenotype. Eventually, the population settles into an equilibrium where back-and-forth epimutations of type-0 cells and reversions of type-1 cells drive an expansion of both subpopulations, albeit at a slow rate.
Once the type-1 population has reached a sufficient size, we can expect it to develop more permanent resistance mechanisms, as is discussed further in  subsequent sections.
In Figure \ref{fig:criticalregion_expbehavior}c, we show long-term expected behavior under an alternative parameter regime, where $|\lambda_0|$ is reduced so that it is of the same order as $\lambda_1$ ($\lambda_1 \sim |\lambda_0|)$. In this regime, type-0 decay is less rapid under anti-cancer treatment, which implies both that the type-1 population builds up more quickly, and that more type-0 cells remain at equilibrium. Note that the equilibrium proportion between type-1 and type-0 cells can in general be computed using expression \eqref{eq:gamma} above.

By examining cell behavior at the individual level, we observe that each type-0 cell is almost guaranteed to go extinct in the baseline regime (it survives with probability 0.005\%), while each type-1 cell survives with 4.9\% probability (Equation \ref{eq:extinctionprob}).
Despite these odds, the anti-cancer agent is unsuccessful in eradicating the drug-sensitive population due to the dynamic switching between phenotypes. In fact, the population as a whole is guaranteed to survive treatment (Equation \ref{eq:tumorsurvival}), which is due to the substantial buffer of type-1 cells that accumulates through type-0 epimutations at the start of treatment and protects the somewhat fragile type-1 population against extinction.
These survival probabilities are generally robust to significant variation in $r_0$, $d_0$, $r_1$ and $d_1$, as we display in Table \ref{table:sensitivity1} in Appendix \ref{app:robustness}. The main exception occurs
when $r_1$ and $d_1$ are changed so that the population becomes subcritical, in which case there is no chance of tumor survival. This transition from guaranteed tumor survival to guaranteed extinction can be quite abrupt, as we explore further in Section \ref{sec:tumorsurvival}. 

The above example allows us to glean two important insights: First of all, the relatively rapid adoption of an epigenetically-mediated resistant phenotype places less restrictions on the robustness of such a phenotype than what is the case for a rare genetic variant arising through mutation.
Thus, even a barely viable non-genetic phenotype may allow the tumor population to escape anti-cancer therapy with 100\% probability, in the absence of any more permanent resistance mechanisms.
Secondly, a tumor population that appears to be static or slow-growing at the population level  may in fact be driven be rapid switching dynamics at the single cell level, and uncovering the exact dynamics may be crucial to understanding how the population responds to treatment, which is the subject of our next section.

\subsection{Tumor survival when switching dynamics are perturbed}
\label{sec:tumorsurvival}
Targeted epigenetic agents, e.g.~inhibitors of DNA methylation and histone deacetylation, have been considered both as a means of reversing the tumorigenic potential of cancer cells and of resensitizing resistant cells to anti-cancer therapy
(see e.g.~\cite{momparler2001potential, flatmark2006radiosensitization, gore2006combined,  juergens2011combination, bhatla2012epigenetic}). In this section, we examine how the probability of tumor survival under continuous anti-cancer treatment depends on the rate of epimutation ($\mu$) and reversion ($\nu$). We then extract insights into the potential benefits of a joint application 
of an anti-cancer agent, aimed at killing the tumor bulk,
and an epigenetic drug, aimed at disrupting the phenotypic switching dynamics.

The probability of survival of the tumor cell population, derived in Appendix \ref{app:extinct}, is shown in Figure \ref{fig:asymmetry}a as a function of $\mu$ and $\nu$.  We consider first the case where no transiently resistant cell preexists treatment ($m=0$) and permanent resistance mechanisms are absent ($\eta=0$ and $\xi_0=\xi_1=0$).
We observe transitions in the dynamics around threshold values of $\mu' \sim 10^{-6}$ and $\nu' \sim 10^{-3}$ per hour. 
The latter threshold reflects a regime change from supercriticality to subcriticality, since the baseline value for the net birth rate $\lambda_1$ is of order $10^{-3}$ per hour (the population is supercritical below the 'critical curve' in Fig \ref{fig:asymmetry}a; see Section \ref{sec:solelyswitching}). When the population is subcritical, the tumor goes extinct with 100\% probability, since the high reversion rate $\nu$ to sensitivity makes it impossible for type-1 cells to expand under treatment. In the supercritical regime, there is always some positive probability that the tumor survives, although this probability will be small for low epimutation rates ($\mu \ll 10^{-6}$). For example, the tumor survival probability corresponding to $\mu = 10^{-10}$ and $\nu = 10^{-4}$ in Figure \ref{fig:asymmetry}a is $0.017\%$, since in this case, epimutations are so infrequent that the type-1 state is unlikely to be adopted by any type-0 cell before the population goes extinct. For high epimutation rates ($\mu \gg 10^{-6}$), however, the type-1 buffer that accumulates at the start of treatment becomes so large that tumor survival is guaranteed whenever the population is supercritical, while extinction is guaranteed whenever the population is subcritical.

The threshold value $\mu' \sim 10^{-6}$ per hour represents the minimal epimutation rate at which tumor survival is certain, given a supercritical population.
In Appendix \ref{app:threshold}, we show that the epimutation rate at which the survival probability is at least $1-u$ ($u \ll 1$) is \vspace*{-12pt}

\begin{align} \label{eq:threshold}
\mu' \approx \frac{|\lambda_0| \log u}{n \log(d_1/r_1)}.
\end{align}
The threshold value $\mu'$ thus depends on the drug-sensitivity $\lambda_0$ of type-0 cells, the robustness of the resistant phenotype ($d_1/r_1$ is the extinction probability of a type-1 clone, assuming no epimutations or reversions, see e.g.~\cite{durrett2015branching}), and the initial population size $n$. 
An order of magnitude change in $\lambda_0$ or $n$ will result in an order of magnitude change in $\mu'$, and in Table \ref{table:sensitivity2} in Appendix \ref{app:robustness}, we show some examples of sensitivity to $r_1$ and $d_1$.
For the baseline parameter regime and $u=0.001$, expression (\ref{eq:threshold}) yields $\mu' \approx 3.59 \cdot 10^{-6}$ per hour, compared to a type-0 birth rate of 
$r_0=0.04$ per hour.
This implies that if the resistant phenotype is adopted once in every 10,000 cell divisions during treatment, the tumor is guaranteed to survive, even if no resistant cell preexists in the population. For larger tumor sizes, e.g.~$n= 10^8$ or $n=10^{10}$ cells, the required adoption rate lowers to once every $10^6$ and once every $10^8$ cell divisions, respectively, which are frequencies typical of resistance-conferring mutations (Appendix \ref{app:parametrization}). Since epigenetic modifications and other non-genetic mechanisms can operate much faster, the above discussion implies that the mere possibility of non-genetically conferred resistance can all but guarantee its emergence,
especially when the tumor is large at detection.

\begin{figure}[t!]
\centering
\hspace*{-1cm} \includegraphics[scale=1] {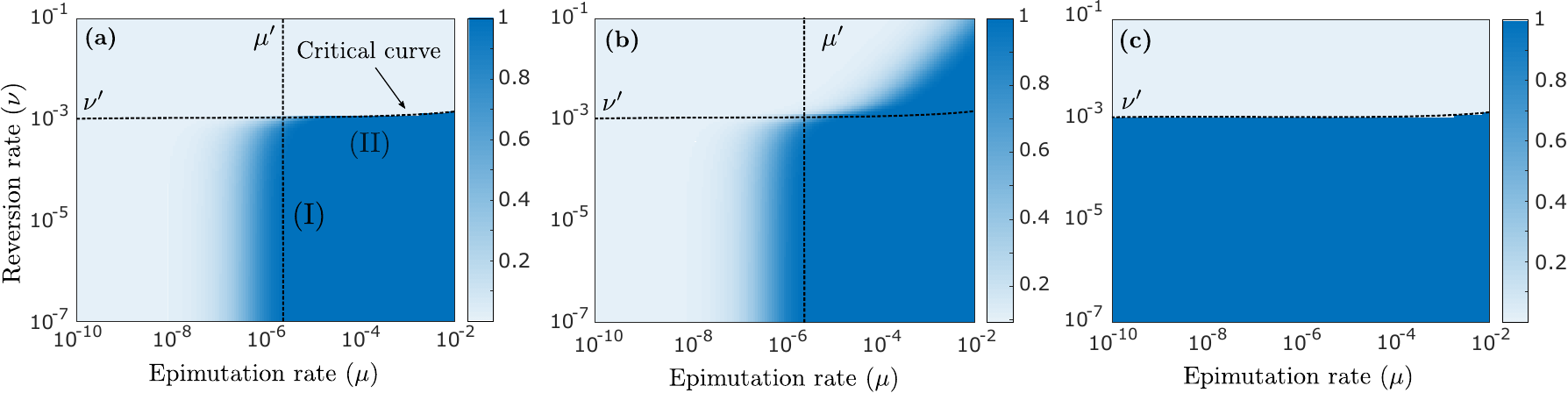}
\caption{
{\bf (a)} Probability that the tumor survives continuous anti-cancer therapy, as a function of $\mu$ and $\nu$, assuming $n=10^6$ and $m=0$ and the absence of permanent resistance mechanisms ($\eta=0$ and $\xi_0=\xi_1=0$), calculated using (\ref{eq:tumorsurvival}). The black dotted curve ('critical curve') separates the regions in the $(\mu,\nu)$ plane where the population is subcritical (top) and supercritical (bottom).
This curve can be extracted from expression \eqref{eq:criticalitycurve} upon logarithmic scaling (see Fig \ref{fig:criticalregion_expbehavior}a). Region (I) in the figure indicates a parameter regime where inhibiting epimutations (decreasing $\mu$) with an epigenetic drug can be an effective treatment strategy, whereas inducing reversions (increasing $\nu$) does little. The reverse is true in the parameter regime of region (II), where a slight perturbation to the reversion rate $\nu$ can guarantee eradication of the tumor cell population.
{\bf (b)} Same as (a), now assuming $\eta=4 \cdot 10^{-7}$ (per hour) and $\xi_0=\xi_1=10^{-7}$ (per cell division), calculated using (\ref{eq:survivalprobperm}).
{\bf (c)} Same as (a), now assuming $n = 10^6 \cdot 0.999$ and $m = 10^6 \cdot 0.001$, calculated using (\ref{eq:tumorsurvival}).
Other parameter values are $r_0 = 0.04$ (per hour), $d_0 = 0.08$, $r_1 = 0.0162$, $d_1 = 0.015$, $\mu = 4 \cdot 10^{-5}$ and $\nu = 4 \cdot 10^{-4}$.
}
\label{fig:asymmetry}
\end{figure}

It is worth noting that the transition from guaranteed tumor survival to guaranteed extinction in Figure \ref{fig:asymmetry}a is much more gradual around the threshold value $\mu' \sim 10^{-6}$ per hour on the $\mu$-axis than around the critical value $\nu' \sim 10^{-3}$ on the $\nu$-axis. This reflects the asymmetric role of type-0 and type-1 cells, and of epimutations and reversions, in the evolutionary dynamics. 
Lowering the epimutation rate $\mu$ will reduce the type-1 buffer that accumulates at the start of treatment, which gradually impairs the collective survival prospects of type-1 cells. On the other hand, since any reversion from type-1 to type-0 effectively amounts to cell death in our setting, increasing the reversion rate $\nu$ will directly affect the survival prospects of individual type-1 cells. As long as $\nu$ is smaller than the critical value, each type-1 cell has some positive probability of persisting therapy, and given a sufficiently large type-1 buffer (i.e.~sufficiently high epimutation rate $\mu$), the tumor as a whole can be guaranteed to survive. Once $\nu$ increases above the critical value, however, each individual type-1 cell becomes certain to go extinct, and the same goes for the tumor as a whole, no matter how large the buffer is.
This explains why for high epimutation rates ($\mu \gg 10^{-6}$), we observe such a sharp transition between guaranteed tumor survival and guaranteed extinction  across the critical curve for $\nu$.

Identifying where a particular cancer cell population falls within the $(\mu,\nu)$ parameter space can yield important insights into the relative attractiveness of targeting $\mu$ and $\nu$ with an epigenetic drug, and the degree to which these parameters should be perturbed.
When $\mu \sim 10^{-6}$ and $\nu \ll 10^{-3}$ per hour, for example (region (I) in Fig \ref{fig:asymmetry}a), inhibiting epimutations (reducing $\mu$) may significantly reduce the tumor survival probability, while inducing reversions (increasing $\nu$) may accomplish little. When $\mu \gg 10^{-6}$ and $\nu \sim 10^{-3}$ per hour, however (region (II) in Fig \ref{fig:asymmetry}a), a slight perturbation to the reversion rate may be the difference between certain tumor survival and certain extinction. This suggests that even if no resistant cell preexists treatment, it may be more effective to revert resistant cells created during the initial stages of therapy than to prevent their emergence.
We also note the importance of identifying the relationship between $\nu$ and $\lambda_1$ for therapeutic considerations. Indeed, recognizing that a small perturbation to the average retention of the resistant phenotype may yield significant treatment benefits can help minimize the risk of any unwanted side effects of the epigenetic treatment.

For the case where permanent resistance mechanisms are assumed ($\eta>0$ and $\xi_0,\xi_1>0$), the probability of tumor survival can be derived by solving a system of nonlinear equations as shown in Appendix \ref{app:stableresistance}.
Figure \ref{fig:asymmetry}b shows  that in this case, the tumor can survive even if the population of type-0 and type-1 cells is subcritical. Indeed, if the epimutation rate $\mu$ is sufficiently high, the large type-1 buffer created at the start of treatment may allow even a subcritical type-1 phenotype to hold off extinction long enough for bona fide resistance to develop.

If we assume that a significant number of resistant cells preexists therapy ($m \gg 0$), the survival probability becomes insensitive to changes in $\mu$ under treatment, since the type-1 buffer needed to save the tumor from extinction will already be present at the onset (Fig \ref{fig:asymmetry}c).
It remains true, however, that a small perturbation to the reversion rate can turn certain therapy failure into certain success.

\begin{figure*}[t!]
\centering
\includegraphics[scale=1]{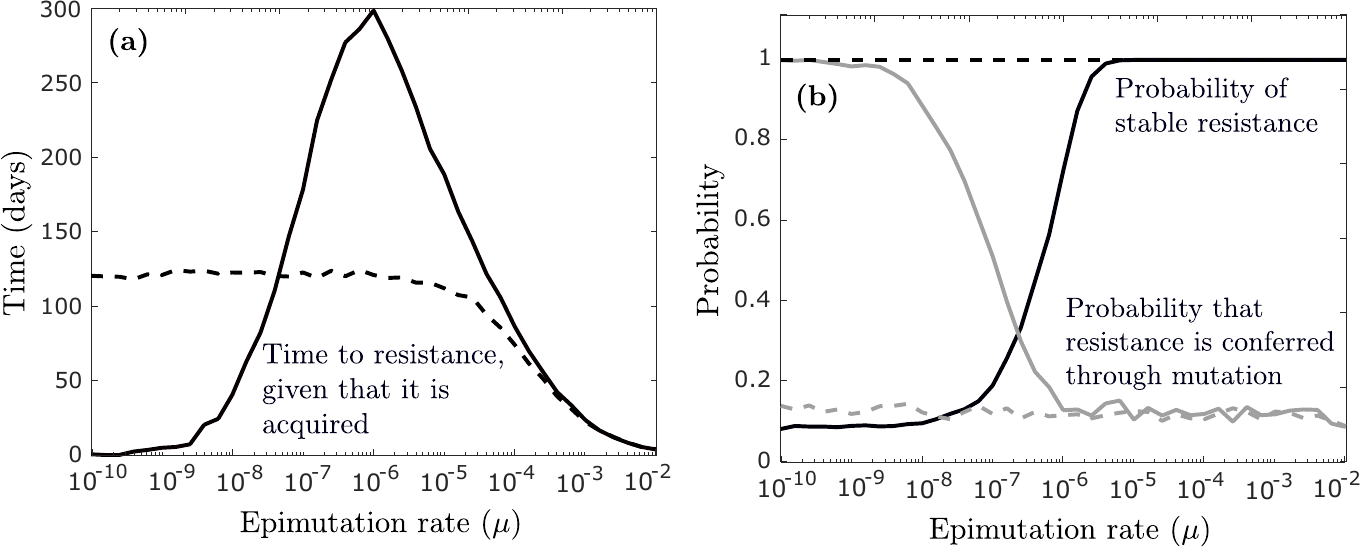}  \vspace*{-6pt}
\caption{
{\bf (a)} Expected time at which the first type-2 cell emerges in the population as a function of the epimutation rate $\mu$, conditioned on the event that a type-2 cell emerges before extinction, first assuming $n=10^6$ and $m=0$ (solid curve), and then $n = 10^6 \cdot 0.999$ and $m = 10^6 \cdot 0.001$ (dashed curve). Produced via simulation by running the process until a type-2 cell emerged on 1000 different occasions and calculating an average.
{\bf (b)} Probability of type-2 emergence (dark curves) and probability that the first type-2 cell arises through mutation (light curves).
Other parameter values are $r_0 = 0.04$ (per hour), $d_0 = 0.08$, $r_1 = 0.0162$, $d_1 = 0.015$, $\nu = 4 \cdot 10^{-4}$, $\eta = 4 \cdot 10^{-8}$ and $\xi_0=\xi_1=10^{-7}$ (per cell division).
}
\label{fig:timetoresistance_mu} \vspace*{-6pt}
\end{figure*}

\subsection{Pathway to stable resistance and rate of acquisition} \label{sec:pathway}
We now examine the mode and speed of stable resistance acquisition ($\eta>0$ and $\xi_0,\xi_1>0$) during continuous anti-cancer treatment. 
In Figure \ref{fig:timetoresistance_mu}a, we show the average time at which the first stably resistant (type-2) cell arises in the population, 
as a function of the epimutation rate $\mu$,
given that a type-2 cell emerges before extinction of the population.
In Figure \ref{fig:timetoresistance_mu}b, we show the probability of the conditioning event, and the probability that stable resistance is conferred through mutation as opposed to epigenetic reprogramming.

We first assume that no transiently resistant (type-1)  cell is present at the onset ($m=0$; solid curves).
We again note a transition in the dynamics around a threshold value of $\mu \sim 10^{-6}$ per hour, and interestingly, 
the average time to resistance is both non-monotonic and highly variable in $\mu$ (Fig \ref{fig:timetoresistance_mu}a; solid curve). When $\mu$ is small ($\mu \ll 10^{-6}$), epimutations are so infrequent that the tumor can only survive via mutation in the type-0 population.
Since type-0 cells decay at a fast rate, any such mutation has to occur early if it is to occur at all.
As $\mu$ increases, the burden of saving the tumor from extinction moves increasingly from type-0 to type-1 cells
(Fig \ref{fig:timetoresistance_mu}b; light solid curve), 
which creates a new pathway for resistance acquisition at a later time.
Once the role of type-0 and type-1 cells in conferring resistance becomes stabilized, however,
the mean acquisition time starts to decrease as $\mu$ increases, since
the type-1 population is created earlier and in greater numbers.
If a significant number of type-1 cells is present at the onset ($m \gg 0$), resistance will be guaranteed to form independently of the epimutation rate under treatment (Fig \ref{fig:timetoresistance_mu}b; dark dashed curve),
and the expected acquisition rate will be insensitive to $\mu$ for small $\mu$ (Fig \ref{fig:timetoresistance_mu}a; dashed curve),
since resistance will most likely develop through
the type-1 population that preexists treatment (Fig \ref{fig:timetoresistance_mu}b; light dashed curve).

In Figure \ref{fig:timetoresistance_nu}a, we consider time to resistance
as a function of the reversion rate $\nu$ instead.
We again observe non-monotonicity and significant variability, with resistance development being slowest around the critical value $\nu \sim 10^{-3}$ per hour. In this case, however, the rate of resistance acquisition is generally insensitive to changes in $\nu$ for 
small $\nu$,
and the transition around the critical value $\nu' \sim 10^{-3}$ is much sharper than around the threshold value $\mu' \sim 10^{-6}$ in Figures \ref{fig:timetoresistance_mu}a and \ref{fig:timetoresistance_mu}b. 
For small $\nu$, stable resistance is mediated primarily through epigenetic reprogramming (Fig \ref{fig:timetoresistance_nu}b; light solid curve), so time to resistance is governed by the size of the type-1 population. Since the net proliferation rate $\lambda_1-\nu$ determines net growth of type-1 cells, the rate of resistance acquisition is not affected by $\nu$ as long as $\nu \ll \lambda_1$. As $\nu$ approaches $\lambda_1$, however, the net proliferation rate $\lambda_1-\nu$ approaches zero, and epigenetic reprogramming slows down considerably. As $\nu$ increases above $\lambda_1$, the population becomes subcritical, and the probability of tumor survival decreases abruptly (Fig \ref{fig:timetoresistance_nu}b; dark solid curve), similarly to what we observed in Section \ref{sec:tumorsurvival}. During this sharp transition, the burden of saving the tumor from extinction moves from the type-1 to the type-0 population (Fig \ref{fig:timetoresistance_nu}b; light solid curve), and since any mutation in the type-0 population must occur early if it is to occur at all, the time to resistance also decreases abruptly as $\nu$ enters the subcritical regime.
Contrary to Figures \ref{fig:timetoresistance_mu}a and \ref{fig:timetoresistance_mu}b, we now observe the same qualitative behavior for the $m=0$ (solid curves) and $m \gg 0$ (dashed curves) cases,
which mirrors our discussion from Section \ref{sec:tumorsurvival}.

The above observations
add an interesting layer to our earlier discussion on the effect of manipulating $\mu$ and $\nu$ on treatment outcome. As an example, when $\mu \gg 10^{-6}$ per hour, inhibiting epimutations (reducing $\mu$) may not do much to prevent acquired resistance,
but it can slow it down considerably.
Moreover, any perturbation to the reversion rate $\nu$ when it is around its critical value may dramatically affect the rate of resistance acquisition.
Note that when $\mu$ is around its threshold value of $\mu \sim 10^{-6}$, resistance will be expected to arise extremely late, while still being guaranteed to develop.
The former is a consequence of the slow-cycling nature of the resistant phenotype, while the latter is a consequence of its rapid adoption under treatment.
This is an important feature of epigenetically-mediated resistance, and this behavior stands in
stark contrast to more robust genetically-resistant phenotypes that are adopted less frequently.

\begin{figure}[t!]
\centering
\includegraphics[scale=1] {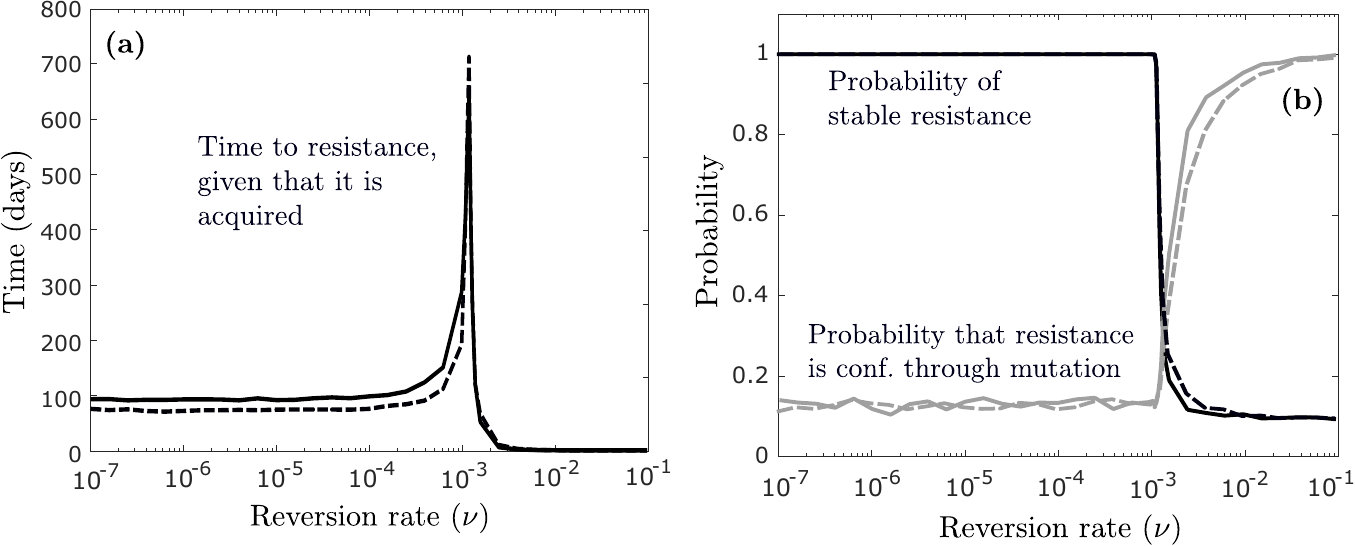} \vspace*{-6pt}
\caption{
{\bf (a)} Expected time at which the first type-2 cell emerges in the population as a function of the reversion rate $\nu$, conditioned on the event that a type-2 cell emerges before extinction, first assuming $n=10^6$ and $m=0$ (solid curve), and then $n = 10^6 \cdot 0.999$ and $m = 10^6 \cdot 0.001$ (dashed curve). Produced via simulation by running the process until a type-2 cell emerged on 1000 different occasions and calculating an average.
{\bf (b)} Probability of type-2 emergence (dark curves) and probability that the first type-2 cell arises through mutation (light curves).
Other parameter values are $r_0 = 0.04$ (per hour), $d_0 = 0.08$, $r_1 = 0.0162$, $d_1 = 0.015$, $\mu = 4 \cdot 10^{-5}$, $\eta = 4 \cdot 10^{-8}$ and $\xi_0=\xi_1=10^{-7}$ (per cell division).
}
\label{fig:timetoresistance_nu}
\end{figure}

\subsection{Evaluation of combination treatment strategies} \label{sec:combination}

To further illustrate how our model can aid medical decision-making, we now examine a combination treatment of an anti-cancer agent and a hypothetical epigenetic drug that directly targets the rates of epimutation ($\mu$) and reversion ($\nu$). Our main questions are whether such a combination is likely to be effective, and whether the epigenetic drug should be applied as pretreatment, posttreatment, or simultaneously with the anti-cancer agent.
To evaluate treatment outcome, we use the following expression for the probability that a successful type-2 cell (i.e.~a type-2 cell that gives rise to a clone that does not go extinct) has emerged by time $t$, \vspace*{-12pt}

\begin{align} \label{eq:timetostableresistance}
1-\exp\left(-\frac{\lambda_2}{r_2} \int_0^t (\xi_0 r_0 \phi_0(s)+\xi_1r_1\phi_1(s) + \eta \phi_1(s)) ds\right)
\end{align}
derived in Appendix \ref{app:stableresistance}. Using (\ref{eq:meanindividual}), we can derive explicit expressions for the integrals in (\ref{eq:timetostableresistance}) as follows: \vspace*{-6pt}
\begin{align}
\begin{split}
\int_0^t \phi_0(s) ds &= -\frac\beta\sigma (1-e^{\sigma t}) - \frac\alpha\rho (1-e^{\rho t}), \\
\eqrevision{\int_0^t \phi_1(s) ds}{\int_1^t \phi_1(s) ds} &= -\frac{\gamma \beta}\sigma (1-e^{\sigma t}) + \frac{\delta\alpha}\rho (1-e^{\rho t}),
\end{split}
\end{align}
assuming $\sigma \neq 0$ and $\rho \neq 0$. These two integrals can be interpreted as the average 'total mass' of type-0 and type-1 cells, respectively, up until time $t$.

We assume that each treatment cycle consists of three two-day blocks and we examine four schedules.
In Schedule A, the epigenetic drug is applied as pretreatment to the anti-cancer agent, whereas in Schedules B, C and D, the anti-cancer agent is applied during the first two blocks of each cycle, and the epigenetic drug is applied during the first, second or third block, respectively (Fig \ref{fig:interval}a).
We show results assuming that the epigenetic drug increases the reversion rate $\nu$ by two orders of magnitude and decreases the epimutation rate $\mu$ to the same extent, while noting that our insights are robust to significant variation in this assumption (Appendix \ref{app:robustness}). We also note that the relatively short duration of each treatment cycle is a function of the rapid {\em in vitro} dynamics of our baseline parameter regime. If each rate constant in the baseline regime is reduced by an order of magnitude, which
may more accurately represent {\em in vivo} dynamics, the results shown below will continue to hold unchanged if we use 60-day treatment cycles instead of 6-day cycles.

\begin{figure}[!t]
\centering
\hspace*{-1cm} \includegraphics[scale=1] {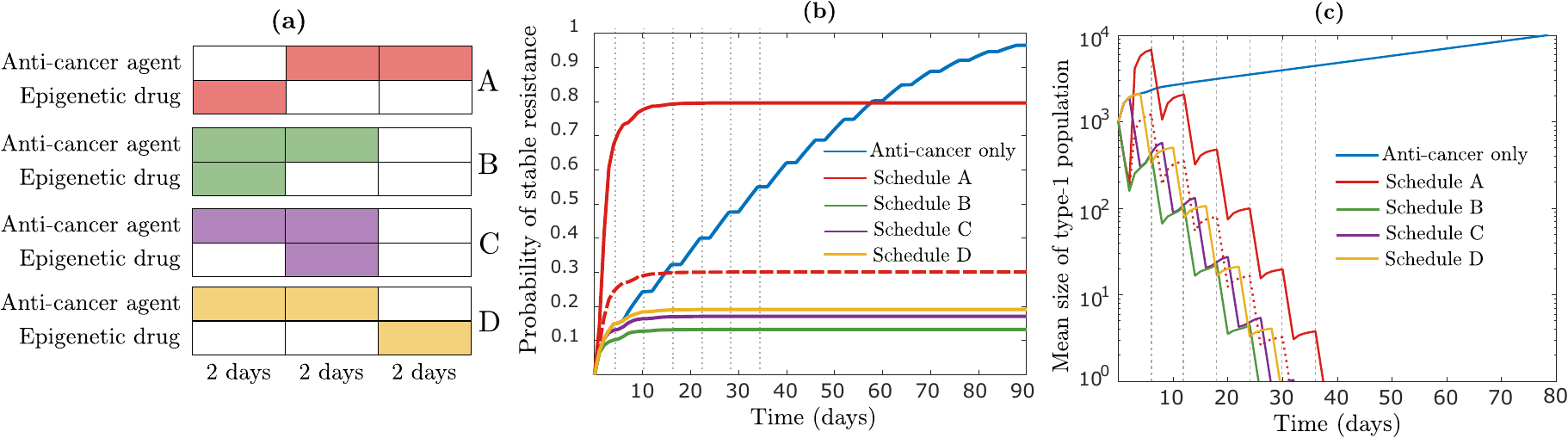}
\caption{
{\bf (a)} Schematic representation of one treatment cycle in the schedules considered. Each schedule is identified by a distinct color which is also applied in (b) and (c).
{\bf (b)} Time-evolution of the probability of successful type-2 emergence, first when the anti-cancer agent is applied alone during the first two blocks of each cycle (blue curve), and then for the four schedules depicted in (a), calculated using (\ref{eq:timetostableresistance}).
Under the anti-cancer agent, $r_0 = 0.04$ (per hour) and $d_0=0.08$, while off it, $r_0 = 0.04$ and $d_0 = 0.0015$. Under the epigenetic drug, $\mu = 4 \cdot 10^{-7}$ and $\nu = 4 \cdot 10^{-2}$, while off it, $\mu = 4 \cdot 10^{-5}$ and $\nu = 4 \cdot 10^{-4}$. The dashed red line shows a modification of Schedule A where $r_0=d_0=0.04$ during the first 2-day block of the first treatment cycle.
{\bf (c)} Time-evolution of the expected number of type-1 cells for the same schedules as considered in (b), calculated using (\ref{eq:meanindividual}).
Parameters not mentioned above are fixed at $r_1 = 0.0162$ (per hour), $d_1 = 0.015$, $\eta = 4 \cdot 10^{-7}$, $\xi_0 = \xi_1= 10^{-7}$ (per cell division), $n = 10^6 \cdot 0.999$ and $m = 10^6 \cdot 0.001$. The vertical broken lines in (b) and (c) signify treatment cycles. 
}
\label{fig:interval}
\end{figure}

Figure \ref{fig:interval}b indicates that whereas the anti-cancer agent alone is guaranteed to fail,
combining it with 
a drug that disturbs the epigenetic switching dynamics 
has the potential to
eradicate the tumor with high probability.
Pretreatment with the epigenetic drug performs the worst by far (Schedule A;  red solid curve), which is due to the fast proliferation of type-0 cells in the absence of the anti-cancer agent, and the relatively fast acquisition of the resistant phenotype in its presence. This both increases the risk of a resistance-conferring mutation during the initial pretreatment phase, and ensures that type-1 cells killed during this phase are rapidly replenished once the anti-cancer agent is applied (Fig \ref{fig:interval}c).
Schedule A continues to perform the worst even if we assume that the tumor population cannot expand beyond its initial size due to spatial constraints, which we model by setting $r_0=d_0=0.04$ per hour during the initial pretreatment phase
(red dashed curves).
This implies that pretreatment with an epigenetic drug may not be judicious when drug-sensitive type-0 cells proliferate too rapidly in the absence of the anti-cancer agent and/or when the resistant type-1 phenotype is adopted too quickly under treatment.
This insight is consistent with experimental findings by Sharma et al.~\cite{sharma2010chromatin},
where the authors observe that simultaneous application of erlotinib and the histone deacetylase inhibitor TSA to PC9 cells eliminates the emergence of resistant clones, whereas a substantial number of such clones arises when extended pretreatment with TSA is followed by exposure to erlotinib.

Focusing now on the other three schedules (B, C and D),
we note that the number of type-1 cells remaining at the end of each treatment cycle is similar for each of them (Fig \ref{fig:interval}c).
Schedule B performs the best, however, since it contains the type-1 population as early as possible in each cycle, which reduces the 'total mass' of type-1 cells in the presence of the anti-cancer agent and slows down epigenetic reprogramming by (\ref{eq:timetostableresistance}).
This finding is robust to significant perturbation of main model parameters, and the relative attractiveness of Schedule B over the other schedules can be much more pronounced than shown here, e.g.~when the rate of epimutation ($\mu$) or epigenetic reprogramming ($\eta$) is higher, or the initial tumor size ($n+m$) is larger (Appendix \ref{app:robustness}).
As an example, when $n= 10^8 \cdot 0.999$ and $m=10^8 \cdot 0.001$, with other parameters unchanged, resistance will be all but guaranteed to form in Schedules A, C and D, while Schedule B is capable of preventing resistance via epigenetic reprogrammming with high probability (Fig \ref{fig:Treatment_appendix1}d).

Our results indicate both that combining an epigenetic drug with an anti-cancer agent can vastly outperform anti-cancer only treatment, and that the epigenetic drug should be applied simultaneously with the anti-cancer agent. These insights are consistent with experimental findings e.g.~by Su et al.~\cite{su2017single}, who find that {\em in vitro} inhibition of the Nf$\kappa$B p65 and MEK/ERK signaling axes in melanoma, in combination with the anti-cancer agent vemurafenib, significantly outperforms monotherapy, by arresting cell-state transitions to a drug-tolerant state.
In Goldman et al.~\cite{goldman2015temporally}, the authors find that applying the SFK inhibitor dasatinib {\em in vivo} on days 8-11 of breast cancer treatment, following two maximum tolerated doses of the anti-cancer agent docetaxel (DTX) on days 2 and 5, is much more effective than DTX-alone treatment. They further show that this treatment is preferable to applying dasatinib on days 2-5 (Schedule I), which shows modest improvement over stand-alone DTX treatment, and applying dasatinib on days 14-17 (Schedule II), which shows no improvement over stand-alone treatment. Although the authors in \cite{goldman2015temporally} refer to Schedule I as a 'simultaneous schedule', the application of dasatinib on days 2-5 should not necessarily be expected to inhibit epimutations following the large DTX dose on day 5.
Application on days 8-11 may be better timed both to revert already drug-tolerant cells to sensitivity and to 'target the induction phase of DTX-induced cell state transition' as the authors put it. Nevertheless, this example serves as an important reminder that actual {\em in vivo} dynamics are likely to be more complex than captured by our simple abstraction.
As an example, if the epimutation rate $\mu$ is influenced by anti-cancer treatment, but this effect is delayed, there may be reason to delay application of the epigenetic drug.
Any substantial presence of drug-tolerant cells that preexist treatment will then create a trade-off between attacking these two temporally distinct sources of resistance.
Indeed, a more complete mechanistic and quantitative understanding of the dynamics of phenotypic switching at the single-cell-level, and how these dynamics are influenced by anti-cancer treatment, will give rise to more complex mathematical questions, as we address further in the discussion section.

\section{Discussion} 
The pervasiveness of acquired drug resistance remains one of the major challenges during clinical management of cancer patients. 
We have established through our evolutionary modeling that
non-heritable stochastic fluctuations in gene expression and short-term epigenetic modifications can 'save' a tumor from extinction
in the absence of any more permanent resistance mechanisms. We have also seen that the potential for rapid adoption of non-genetic resistance implies that such resistance can be all but guaranteed to develop,
even if no resistant cell preexists treatment
and the resistant phenotype is barely viable. 
This suggests ample opportunity for non-genetic mechanisms to induce failure of conventional anti-cancer therapy, and it may help explain 
why so many tumors develop resistance without acquiring mutations in drug targets or activated pathways
\cite{wilting2012epigenetic}.

Throughout, we have highlighted how thinking in terms of our mathematical model can aid medical decision-making. In Section \ref{sec:tumorsurvival}, we saw how a quantitative understanding of the dynamics of epimutation and reversion at the single-cell level can yield important insights into which parameter to attack with an epigenetic drug and the extent to which it needs to be perturbed. As an example, when
the average retention $1/\nu$ of the resistant phenotype is near a critical value of $1/\lambda_1$, a slight perturbation to $\nu$ can turn guaranteed treatment failure into guaranteed success.
In Section \ref{sec:pathway}, we further saw how a quantitative understanding of the underlying switching dynamics, 
and the rate at which any permanent resistance mechanisms are adopted,
can allow us to predict the mode and time scale of resistance acquisition.
We noted that when the resistant type-1 phenotype is slow-cycling, the expected time at which permanent resistance arises in the population can vary greatly depending on the epimutation rate $\mu$, and that when $\mu$ is near a certain threshold value, resistance can be expected to develop extremely late while still being guaranteed to emerge.

In Section \ref{sec:combination}, we finally observed that combining an epigenetic drug with an anti-cancer agent can significantly outperform anti-cancer only treatment, and that the epigenetic drug should be applied simultaneously with the anti-cancer agent rather than as pretreatment or posttreatment. We also noted that pretreatment with the epigenetic drug is not advisable especially when drug-sensitive cells proliferate too rapidly in the absence of the anti-cancer agent, or the resistant phenotype is adopted too quickly in its presence. It should of course be stressed that our model is highly simplified. As an example, we have not assumed any delay in the potential drug-induced adoption of transient resistance under anti-cancer treatment, and we have not assumed any interaction between the anti-cancer agent and the epigenetic drug. Also, whereas we have assumed that the drug-resistant phenotype is continuously slow-cycling, it may respond to drug pressure by entering a quiescent state before resuming proliferation. Finally, our model does not incorporate spatial effects or toxicity constraints, all of which may call for modifications or extensions to the model as our biological understanding accumulates.

Our analysis represents, as far as we know, a first attempt toward a deeper understanding of the evolutionary dynamics of a population that is able to employ transiently resistant cell states to escape anti-cancer therapy. The model presented is flexible in that it enables description of a variety of sources of phenotypic heterogeneity, and the analytical expressions we have derived allow for easy estimation and comparison of treatment outcomes under different regimens. Our investigation is theoretical in nature, however, and the true power of our mathematical model will only be realized through a more complete mechanistic understanding of the dynamics of phenotypic switching at the single-cell-level, and through new methods to infer these dynamics quantitatively. Our hope is that a combined biological and mathematical effort will pave the way toward the design of novel therapeutic strategies, as well as mathematical optimization of already available drug combinations and treatment schedules \cite{chmielecki2011optimization,leder2014mathematical}.

\section*{Acknowledgements}
EBG and KL were supported in part by NSF grant CMMI-1362236. EBG and JF were supported in part by NSF grant DMS-1349724. KL and JF were supported in part by the U.S.-Norway Fulbright Foundation. SD was supported in part by NIH grant R01GM129066.

\section*{Competing Interests}
The authors declare no competing interests.

\bibliographystyle{ieeetr}
\bibliography{epi}











\begin{appendices}

\renewcommand\thefigure{\thesection.\arabic{figure}}
\renewcommand\thetable{\thesection.\arabic{table}}
\renewcommand\theequation{\thesection.\arabic{equation}}
\setcounter{figure}{0} 
\setcounter{equation}{0} 
\setcounter{table}{0} 

\section{Parametrization} \label{app:parametrization}

To parametrize the model, we rely on Sharma et al.'s~\cite{sharma2010chromatin}  investigation of an {\em EGFR}-mutant non-small-cell lung cancer (NSCLC) population (PC9) treated with 2 $\mu$M erlotinib, with fresh drug added every 3 days, and Hata et al.'s~\cite{Engelman2016} study of the evolution of resistance conferred by the $EGFR^{\rm T790M}$ gatekeeper mutation in PC9 cells.

Most of the discussion in the main text concerns behavior in the presence of an anti-cancer agent, so we start with those dynamics. 
Following Hata et al.~\cite{Engelman2016}, we assume that drug-tolerant type-1 cells give birth at rate $r_1 = 0.0162$ per hour and die at rate $d_1 = 0.015$ per hour, and that stably resistant type-2 cells give birth at rate $r_2 = 0.04$ and die at rate $d_2 = 0.0015$.
Although these rates apply to treatment with 300 nM gefitinib as opposed to 2 $\mu$M erlotinib, viability curves presented in \cite{sharma2010chromatin} indicate that it is reasonable to assume similar rates for the case of 2 $\mu$M erlotinib treatment.

In Sharma et al.~\cite{sharma2010chromatin}, the authors report that under 9 days of continuous anti-cancer treatment, almost all cells in the population die, while persister cells corresponding to around $0.27\%  \pm 0.21\%$ of the original population survive. Assuming exponential decay of the drug-sensitive type-0 population, we estimate a net birth rate of $\lambda_0=-0.04$ per hour for type-0 cells and set $r_0 = 0.04$ and $d_0 = 0.08$, where we assume without loss of generality that the anti-cancer agent increases the death rate of type-0 cells without affecting their birth rate (the birth rate of type-0 cells in the absence of the anti-cancer agent is $r_0=0.04$ as is discussed below). Note that if we assume instead that the drug decreases the birth rate of type-0 cells, with the net birth rate $\lambda_0 = -0.04$ unchanged, the type-0 population will decay at the same net rate as before, but the mutation rate $r_0\xi_0$ of type-0 cells will decrease proportionally to the decrease in $r_0$. We can therefore view variation in the effect of the drug on the birth rate $r_0$ as equivalent to variation in the mutation rate $\xi_0$.

The rates $\mu$ of epimutation and $\nu$ of reversion are more difficult to estimate, since it is not clear in \cite{sharma2010chromatin} how many of the type-1 persister cells that survive the first 9 days of treatment are present at the onset. However, the authors do conclude from experiments applying the histone deacetylase inhibitor TSA  as pretreatment to erlotinib that type-1 cells do emerge {\em de novo} during treatment.
For our baseline parameter regime, we will assume that around half the cells that survive the first 9 days of treatment are present at the onset and set $\mu = 4 \cdot 10^{-5}$ per hour accordingly, which implies that an epimutation occurs once in every 1,000 divisions of a type-0 cell.
In determining $\nu$, 
we assume that cells transition more freely out of the resistant state than into it and set $\nu = 4 \cdot 10^{-4}$, which is consistent with observed phenotypic switching dynamics between persister cells and normal cells in {\em Escherichia coli} bacterial populations \cite{norman2015stochastic}, and between stem-like and non-stem-like cells in breast cancer \cite{gupta2011stochastic}. 

The point mutation rate per nucleotide per cell division has been estimated as $5 \cdot 10^{-10}$ \cite{durrett2015branching}. To obtain the rate of mutations that confer resistance to anti-cancer therapy, this number needs to be multiplied by the number of resistance-conferring point mutations. Following e.g.~\cite{Engelman2016} and \cite{bozic2013evolutionary}, we assume a baseline rate of $\xi_0 = \xi_1= 10^{-7}$ per cell division, although a reasonable range can be anywhere from $10^{-9}$ to $10^{-5}$, which is the range suggested by \cite{durrett2015branching} for mutations leading to cancer.
This assumption on $\xi_0$ and $\xi_1$ translates into a mutation rate of $\eqrevision{r_0\xi_0}{r_0\xi} = 4 \cdot 10^{-9}$ per hour for type-0 cells and $\eqrevision{r_1\xi_1}{r_1\xi} = 1.62 \cdot 10^{-9}$ per hour for type-1 cells.

\begin{table}[t]
\centering
\begin{tabular}{| m{1.5cm} | m{3cm} | m{3cm} |}
\hline
&Presence of drug \\
\hline
$r_0$&0.04 \\
\hline
$d_0$&0.08 \\
\hline
$r_1$&0.0162 \\
\hline
$d_1$&0.015 \\
\hline
$\mu$&$4 \cdot 10^{-5}$ \\
\hline
$\nu$&$4 \cdot 10^{-4}$ \\
\hline
$\eta$ & $4 \cdot 10^{-7}$ or $4 \cdot 10^{-8}$ \\
\hline
$\xi_0=\xi_1$ & $10^{-7}$\\
\hline
$r_2$ & $0.04$ \\
\hline
$d_2$ & $0.0015$ \\
\hline
\end{tabular}
\caption{Baseline parameter regime in the presence of an anti-cancer agent. All rates are measured per hour, except the mutation rates $\xi_0$ and $\xi_1$, which are measured per cell division.}
\label{parameters1}
\end{table}

As for the rate $\eta$ of epigenetic reprogramming, we note that measurements provided by Bintu et al.~\cite{bintu2016dynamics} of the dynamics of epigenetic silencing under recruitment of chromatin regulators suggest that 
it may be natural to assume that $\eta$ is around 1-2 orders of magnitude lower than $\mu$.  We will apply $\eta= 4 \cdot 10^{-7}$ or $\eta = 4 \cdot 10^{-8}$ as our baseline rate depending on the context, which is 2-3 orders of magnitude smaller than $\mu$, and 1-2 orders of magnitude larger than the mutation rates $r_0\xi_0$ and $r_1\xi_1$. As is the case for the mutation rate, one can expect significant variation in the epigenetic reprogramming rate depending on the cancer type, the anti-cancer agent being applied and the concentration of this agent.
We also note that 
`stable epigenetic resistance' may not even be a well-defined concept
if epigenetic reprogramming occurs in a multi-stage or continuous fashion, with each stage leading to an increasingly stable phenotype, which is either less likely to revert back to sensitivity once removed from drug or does so on a longer time scale.

In the main text, we show all results assuming an initial population size of $10^6$ cells, following the {\em in vitro} experiments conducted by Sharma et al. We note that $10^6$ refers to the {\em effective population size}, i.e.~those tumor cells that are capable of undergoing phenotypic switching,
which may only apply to a subset of tumor cells~\cite{michor2006evolution}.
We also note that the initial population size should be viewed in context of the mutation rates  $\eqrevision{r_0\xi_0}{r_0\xi} = 4 \cdot 10^{-9}$ and $\eqrevision{r_1\xi_1}{r_1\xi} = 1.62 \cdot 10^{-9}$ and the reprogramming rates $\eta = 4 \cdot 10^{-7}$ or $\eta = 4 \cdot 10^{-8}$, 
since the relationship between these three parameters largely determines the dynamics of resistance acquisition.

The only section in the main text where tumor dynamics in the absence of an anti-cancer agent are considered is Section \ref{sec:combination}. Since most of the data referenced above concerns behavior in the presence of drug, we will mostly assume that the dynamics on and off drug are identical, with the important exception of the birth and death rate of type-0 cells, which we assume to be $r_0 = 0.04$ and $d_0=0.0015$ off drug, following Hata et al.~\cite{Engelman2016}.
The resulting parameter values (Table \ref{parameters2}) therefore reflect the qualitative setting where type-0 cells proliferate rapidly in the absence of drug but die rapidly in its presence, and type-1 cells proliferate slowly both on and off drug. The assumption that type-1 cells are at a selective disadvantage off drug is not essential to our results, as is discussed in Appendix \ref{app:robustness} below.

\begin{table}[h]
\centering
\begin{tabular}{| m{1.5cm} | m{3cm} |}
\hline
&Absence of drug\\
\hline
$r_0$& 0.04 \\
\hline
$d_0$&0.0015 \\
\hline
$r_1$&0.0162 \\
\hline
$d_1$& 0.015 \\
\hline
$\mu$&$4 \cdot 10^{-5}$ \\
\hline
$\nu$& $4 \cdot 10^{-4}$ \\
\hline
$\eta$ & 0 \\
\hline
$\xi_0=\xi_1$ & $10^{-7}$ \\
\hline
$r_2$ & 0.04 \\
\hline
$d_2$ & 0.0015 \\
\hline
\end{tabular}
\caption{Baseline parameter regime in the absence of an anti-cancer agent. All rates are measured per hour, except the mutation rates $\xi_0$ and $\xi_1$, which are measured per cell division.} 
\label{parameters2}
\end{table}

We finally note that according to the above dynamics, type-1 cells will constitute around 0.11\% of the population during long-term expansion in the absence of drug and at the start of treatment. In the main text, we generally consider both the case where no type-1 cell is present at the start of anti-cancer treatment ($m=0$), and the case where type-1 cells constitute 0.1\% of the population at the onset ($m \gg 0$).

\renewcommand\thefigure{\thesection.\arabic{figure}}
\renewcommand\thetable{\thesection.\arabic{table}}
\renewcommand\theequation{\thesection.\arabic{equation}}
\setcounter{figure}{0} 
\setcounter{equation}{0} 
\setcounter{table}{0} 

\section{Two-type Markovian branching processes} \label{app:generalbranching}
Here, we derive expression (\ref{eq:meanindividual}) in the main text
and describe some of its properties.
For the case $\eta=0$ and $\xi_0=\xi_1=0$, our model reduces to a two-type continuous-time Markovian branching process~\cite{athreya2004branching}. Let $X(t) = (X_0(t), X_1(t))$ denote such a process, where the two types are designated as type-0 and type-1. Associated with $X(t)$ is the {\em mean matrix} ${\bf M}(t) = \{m_{ij}(t): (i,j)\in \{0,1\}^2 \}$, defined by  
\[
m_{ij}(t)=E[X_j(t)|X(0)= {\bf e}_i],\quad (i,j)\in \{0,1\}^2,
\]
where ${\bf e}_i$ denotes the unit vector with 1 in the $i$-th coordinate, and the {\em infinitesimal generator} ${\bf A} = \{a_{ij}: (i,j) \in \{0,1\}^2\}$, which satisfies   \vspace*{-12pt}

\begin{align*} 
\mathbf{M}(t)=\exp(\mathbf{A}t), \quad t \geq 0.
\end{align*}
We can interpret $m_{ij}(t)$ as the mean number of type-$j$ particles alive at time $t$, given that the process is started by a single type-$i$ cell, and $a_{ij}$ as the infinitesimal rate at which a type-$i$ cell produces a type-$j$ cell. We assume that $a_{12}, a_{21}>0$, i.e.~the rates of switching between types are strictly positive. For further discussion on the above matrices, see \cite{athreya2004branching}.

Let $\phi_0^{(n,m)}(t)$ and  $\phi_1^{(n,m)}(t)$ denote the mean number of type-0's and type-1's alive at time $t$, assuming initial conditions $(X_0(0),X_1(0)) = (n,m)$, where $n$ and $m$ are nonnegative integers with $n+m>0$.
Using the mean matrix, we can easily compute these means as \vspace*{-12pt}

\begin{align} \label{eq:meannumtypes1}
[\phi_0^{(n,m)}(t) \;\; \phi_1^{(n,m)}(t)]= [n \,\; m] \; {\bf M}(t).
\end{align}
Note that the infinitesimal generator ${\bf A}$ possesses distinct real eigenvalues $\rho < \sigma$ given by \vspace*{-12pt}

\begin{align} \label{eq:eigenvalues}
\frac{a_{11}+a_{22} \pm \sqrt{(a_{11}-a_{22})^2+4a_{12}a_{21}}}2,
\end{align}
which follows from our assumption that $a_{12}>0$ and $a_{21}>0$.
The eigenvalues of ${\bf M}(t) = \exp({\bf A}t)$ are then easily obtained as $e^{\rho t}$ and $e^{\sigma t}$. If we define $\delta := (a_{11}-\rho)/a_{21}$ and $\gamma := (\sigma-a_{11})/a_{21}$,
it is easily established that $\mathbf{v} = [1 \;\,\delta]$ and $\mathbf{w} = [1\;\, \gamma]$
are left eigenvectors of ${\bf A}$ with respect to $\rho$ and $\sigma$, respectively. Decomposing $[n \;\, m]$ in terms of $\mathbf{v}$ and $\mathbf{w}$, we can then compute the means in (\ref{eq:meannumtypes1}) explicitly as \vspace*{-12pt}

\begin{align}
\begin{split} \label{eq:meanexpressions}
\phi_0^{(n,m)}(t)  &= \frac{n\delta + m}{\delta + \gamma} e^{\sigma t} + \frac{n\gamma -m}{\delta+\gamma} e^{\rho t},  \\
\phi_1^{(n,m)}(t) &= \frac{\gamma(n\delta+m)}{\delta+\gamma} e^{\sigma t} - \frac{\delta(n\gamma-m)}{\delta+\gamma} e^{\rho t}.
\end{split}
\end{align}
If we define $\alpha := (n\gamma-m)/(\delta+\gamma)$ and $\beta := (n\delta+m)/(\delta+\gamma)$, we can simplify these expressions further to \vspace*{-12pt}

\begin{align}
\begin{split} \label{eq:meanexpressions1}
\phi_0^{(n,m)}(t)  &= \beta e^{\sigma t} + \alpha e^{\rho t},  \\
\phi_1^{(n,m)}(t) &= \gamma \beta e^{\sigma t} - \delta \alpha e^{\rho t}.
\end{split}
\end{align}
Note that although we do not show it explicitly, the constants $\alpha$ and $\beta$ depend on the initial conditions $(n,m)$, whereas $\gamma$, $\delta$, $\sigma$ and $\rho$ are completely determined by the infinitesimal generator ${\bf A}$.

We next establish that that $a_{11}-\rho >0$ and $\sigma-a_{11} >0$, which will imply that $\gamma>0$, $\delta>0$ and $\beta>0$ in (\ref{eq:meanexpressions}) and (\ref{eq:meanexpressions1}). By rewriting (\ref{eq:meanexpressions}) as \vspace*{-12pt}

\begin{align*}
\phi_0^{(n,m)}(t)  &= \frac{n}{\delta+\gamma} (\delta e^{\sigma t}+\gamma e^{\rho t}) + \frac{m}{\delta+\gamma}(e^{\sigma t} - e^{\rho t}), \\
\phi_1^{(n,m)}(t) &= \frac{n\delta \gamma}{\delta+\gamma}(e^{\sigma t}-e^{\rho t}) + \frac{m}{\delta+\gamma} (\gamma e^{\sigma t}+ \delta e^{\rho t}),
\end{align*}
and noting that $\rho < \sigma$, it will also follow that the expected number of type-0's and type-1's is strictly positive for all $t>0$, i.e.~$X(t)$ is a positive-regular process \cite{athreya2004branching}.

To show that $a_{11}-\rho >0$ and $\sigma-a_{11} >0$, we first note that since $\rho$ is an eigenvalue for ${\bf A}$, we have $(a_{11}-\rho)(a_{22}-\rho) = a_{12}a_{21}$. By our assumption that $a_{12}, a_{21}>0$, the terms $a_{11}-\rho$ and $a_{22}-\rho$ must then be nonzero and have the same sign. Using (\ref{eq:eigenvalues}), we get since $a_{12},a_{21}>0$: \vspace*{-12pt}

\begin{align*}
\rho &= \frac{a_{11}+a_{22} - \sqrt{(a_{11}-a_{22})^2+4a_{12}a_{21}}}2 \\
&< \frac{a_{11}+a_{22} - |a_{11}-a_{22}|}2 \\
&= \begin{cases} a_{11} \text{  if  $a_{11} \leq a_{22}$} \\ a_{22} \text{  if  $a_{11} > a_{22}$.} \end{cases}
\end{align*}
This implies that $a_{11}-\rho>0$ or $a_{22}-\rho>0$ depending on the relationship between $a_{11}$ and $a_{22}$. Since $a_{11}-\rho$ and $a_{22}-\rho$ have the same sign, one being positive implies that the other one is as well, so in both cases, we have $a_{11}-\rho>0$ and $a_{22}-\rho>0$. A similar argument confirms that $\sigma-a_{11} > 0$.

Now let \vspace*{-12pt}

\[
\phi^{(n,m)}(t) := \phi_0^{(n,m)}(t)+\phi_1^{(n,m)}(t)
\]
denote the expected total number of cells (type-0 and type-1) alive at time $t$, starting from $n$ type-0 cells and $m$ type-1 cells.  By (\ref{eq:meanexpressions1}), we can write \vspace*{-12pt}

\begin{align} \label{eq:asymptotics0}
\phi^{(n,m)}(t) = \beta(1+\gamma) e^{\sigma t} + \alpha(1-\delta) e^{\rho t}.
\end{align}
We then obtain the following expressions for the expected size of a clone, at time $t$, started by a single cell of each type: \vspace*{-12pt}

\begin{align} \label{eq:asymptotics00}
\begin{split}
\phi^{(1,0)}(t) &= \frac{\delta(1+\gamma)}{\delta+\gamma} e^{\sigma t} + \frac{\gamma(1-\delta)}{\delta+\gamma} e^{\rho t}, \\
\phi^{(0,1)}(t) &= \frac{1+\gamma}{\delta+\gamma} e^{\sigma t} - \frac{1-\delta}{\delta+\gamma} e^{\rho t}.
\end{split}
\end{align}
We conclude by discussing large-$t$ asymptotics. Since $\sigma > \rho$, the long-run behavior will be dominated by the former term in (\ref{eq:meanexpressions}), (\ref{eq:meanexpressions1}), (\ref{eq:asymptotics0}) and (\ref{eq:asymptotics00}). In particular, we can write \vspace*{-12pt}

\begin{align} \label{eq:asymptotics1}
\begin{split}
&\phi_0^{(n,m)}(t) = \beta e^{\sigma t} + o(e^{\sigma t}), \\
&\phi_1^{(n,m)}(t) = \gamma \beta e^{\sigma t} + o(e^{\sigma t}),
\end{split}
\end{align}
and \vspace*{-12pt}

\begin{align} \label{eq:asymptotics2}
\begin{split}
&\phi^{(1,0)}(t) = \frac{\delta(1+\gamma)}{\delta+\gamma} e^{\sigma t}+ o(e^{\sigma t}), \\
&\phi^{(0,1)}(t) = \frac{1+\gamma}{\delta+\gamma} e^{\sigma t}+ o(e^{\sigma t}),
\end{split}
\end{align}
where $o(e^{\sigma t})$ denotes a function that satisifies $o(e^{\sigma t})/e^{\sigma t} \to 0$ as $t \to \infty$.
Note that from (\ref{eq:asymptotics1}), it is clear that $\gamma$ is the long-run ratio between type-1's and type-0's in the population, and from (\ref{eq:asymptotics2}), we see that $\delta$ is the long-run ratio between mean clone size started by a single type-0 cell vs.~a single type-1 cell. This gives an intuitive interpretation of $\delta$ and $\gamma$ in (\ref{eq:meanexpressions}), (\ref{eq:meanexpressions1}), (\ref{eq:asymptotics0}) and (\ref{eq:asymptotics00}), while $\alpha$ and $\beta$ depend on $n$, $m$, $\delta$ and $\gamma$ through $\alpha = (n\gamma-m)/(\delta+\gamma)$ and $\beta = (n\delta+m)/(\delta+\gamma)$.

\renewcommand\thefigure{\thesection.\arabic{figure}}
\renewcommand\thetable{\thesection.\arabic{table}}
\renewcommand\theequation{\thesection.\arabic{equation}}
\setcounter{figure}{0} 
\setcounter{equation}{0} 
\setcounter{table}{0}

\section{Extinction probabilities} \label{app:extinct}

We continue to assume the absence of permanent resistance mechanisms ($\eta=0$ and $\xi_0=\xi_1=0$), which implies that our model still reduces to a two-type branching process $(X_0(t),X_1(t))$.
Define \vspace*{-12pt}

\begin{align*}
& p_0 := P(X_0(t) +X_1(t) = 0 \text{ for some $t$}  | X_0(0)=1, X_1(0) = 0), \\
& p_1 := P(X_0(t) +X_1(t) = 0 \text{ for some $t$}  | X_0(0)=0, X_1(0) = 1),
\end{align*}
 as the extinction probabilities of clones started by a single cell of each type.
Note that if we start the process with a single type-0 cell, the initial event in the process will be a (i) cell division with probability $r_0/(r_0+d_0+\mu)$, (ii) cell death with probability $d_0/(r_0+d_0+\mu)$, and a switch to type-1 with probability $\mu/(r_0+d_0+\mu)$. If the initial event is a cell division, an additional type-0 cell will be created and the process goes extinct if and only if both cells go extinct, which occurs with probability $p_0^2$. If the initial event is a switch between types, the new type-1 cell will go extinct with probability $p_1$. Therefore, by conditioning on whether the initial event is a division, death or switch between types, we can derive the following conditions for $p_0$ and $p_1$: \vspace*{-12pt}
 
 \begin{align*} 
& p_0 = \frac{r_0}{r_0+d_0+\mu} \cdot p_0^2 + \frac{d_0}{r_0+d_0+\mu} \cdot 1+ \frac\mu{r_0+d_0+\mu} \cdot p_1,
\end{align*}
which yields \vspace*{-12pt}

 \begin{align*} 
& (r_0+d_0+\mu) \cdot p_0 = r_0 p_0^2 + d_0 + \mu p_1.
\end{align*}
By the same argument for a process started with a single type-1 cell, we obtain the following nonlinear system for $p_0$ and $p_1$: \vspace*{-12pt}

 \begin{align} \label{eq:extinctionprob}
 \begin{split}
& (r_0+d_0+\mu) \cdot p_0 = r_0 p_0^2 + d_0 + \mu p_1, \\
& (r_1+d_1+\nu) \cdot p_1 = r_1 p_1^2 + d_1 + \nu p_0.
\end{split}
\end{align}
When $\sigma>0$ and the process is supercritical, we have $p_0<1$ and $p_1<1$ which are uniquely determined by these two equations, while when $\sigma < 0$, we have $p_0=1$ and $p_1=1$ \cite{athreya2004branching}.

Note that the classification into supercritical and subcritical is independent of the initial conditions $(X_0(0),X_1(0)) = (n,m)$ (a change in initial conditions does not affect the extinction probability of a single-cell derived clone). The initial conditions do however affect the survival probability of the population as a whole. Indeed, the survival probability for a population starting with $n$ type-0 cells and $m$ type-1 cells is \vspace*{-12pt}

\begin{align} \label{eq:tumorsurvival}
1-p_0^np_1^m\eqrevision{,}{.}
\end{align}
since the population goes extinct if and only if each individual type-0 cell and each individual type-1 cell goes extinct, which occurs with probability $p_0^np_1^m$ by independence.

\renewcommand\thefigure{\thesection.\arabic{figure}}
\renewcommand\thetable{\thesection.\arabic{table}}
\renewcommand\theequation{\thesection.\arabic{equation}}
\setcounter{figure}{0} 
\setcounter{equation}{0} 
\setcounter{table}{0} 

\section{Conditions for supercriticality} \label{app:lemma1}

Here, we derive condition (\ref{eq:supercriticality}) in the main text. We again let $(X_0(t),X_1(t))$ denote the stochastic process corresponding to the model in the absence of permanent resistance mechanisms ($\eta=0$ and $\xi_0=\xi_1=0$).

\begin{lemma} \label{lemma:criticality}
Assume $\lambda_0<0$, $\lambda_1>0$, $\mu>0$ and $\nu>0$.
\begin{enumerate}[(1)]
\item $(X_0(t), X_1(t))$ is subcritical if $\nu \lambda_0 + \mu \lambda_1 < \lambda_0 \lambda_1$, critical if $\nu \lambda_0 + \mu \lambda_1 = \lambda_0 \lambda_1$ and supercritical if $\nu \lambda_0 + \mu \lambda_1 > \lambda_0 \lambda_1$.
\item A sufficient condition for supercriticality is $\lambda_1 \geq \nu$.
\end{enumerate}
\end{lemma}

\begin{proof}
It is trivial to check, by determining the sign of the larger eigenvalue $\sigma$ in (\ref{eq:eigenvaluesmain}) in the main text, that
\begin{enumerate}[(i)]
\item $(X_0(t), X_1(t))$ is subcritical if $\lambda_0+\lambda_1 < \mu+\nu$ and $\nu \lambda_0 + \mu \lambda_1 < \lambda_0 \lambda_1$. 
\item $(X_0(t), X_1(t))$ is critical if $\lambda_0+\lambda_1 < \mu+\nu$ and $\nu \lambda_0 + \mu \lambda_1 = \lambda_0 \lambda_1$. 
\item $(X_0(t), X_1(t))$ is supercritical if $\lambda_0+\lambda_1 \geq \mu+\nu$ or $\nu \lambda_0 + \mu \lambda_1 > \lambda_0 \lambda_1$. 
\end{enumerate}
What remains to show is that the first condition is redundant in each case, and that (iii) is satisfied when $\lambda_1 \geq \nu$.

Note first that if $\lambda_1 \geq \nu$, we have 
\[
\nu \lambda_0 + \mu \lambda_1 - \lambda_0 \lambda_1 = \lambda_0(\nu-\lambda_1) + \mu \lambda_1 \geq \mu \lambda_1  > 0.
\]
since $\lambda_0 < 0$, $\mu > 0$ and $\lambda_1 > 0$, from which supercriticality follows. This establishes (2). Then note that since $\lambda_1 \geq \nu$ implies $\nu \lambda_0 + \mu \lambda_1 > \lambda_0 \lambda_1$, and we have $\nu-\lambda_0 > 0$, we get 
\[
\nu \lambda_0 + \mu \lambda_1 \leq  \lambda_0 \lambda_1 \;\;\Rightarrow\;\; \lambda_1 < \nu  \;\;\Rightarrow\;\;  \lambda_1 < \mu+\nu-\lambda_0,
\]
which shows that the first condition in (i) and (ii) above is implied by the second condition. Using what we have proven, we additionally have
\[
\lambda_0+\lambda_1 \geq \mu+\nu \;\;\Rightarrow\;\; \lambda_1 \geq \nu \;\;\Rightarrow\;\; \nu \lambda_0 + \mu \lambda_1 - \lambda_0 \lambda_1  > 0,
\]
showing that the second condition in (iii) above contains the first one.
\end{proof}

\renewcommand\thefigure{\thesection.\arabic{figure}}
\renewcommand\thetable{\thesection.\arabic{table}}
\renewcommand\theequation{\thesection.\arabic{equation}}
\setcounter{figure}{0} 
\setcounter{equation}{0} 
\setcounter{table}{0} 

\section{Alternate version of the model} \label{app:alternateversion}

In our baseline model, we assume that epimutations and reversions can occur at any time during the cell cycle. If we instead want a model where these events only occur at cell divisions, we proceed as follows, assuming again the absence of permanent resistance mechanisms ($\eta=0$ and $\xi_0=\xi_1=0$).

Let $u$ denote the probability of epimutation at each division of a type-0 cell and $v$ denote the probability of reversion at each division of a type-1 cell. A type-0 cell then divides into two type-0 cells at rate $r_0(1-u)$ and it produces one type-0 and one type-1 cell at rate $r_0u$. Similarly, a type-1 cell divides into two type-1 cells at rate $r_1(1-v)$ and it produces one type-0 and one type-1 cell at rate $r_1v$. These dynamics can be captured by the infinitesimal generator \vspace*{-12pt}

\begin{align*}
{\bf A} = \begin{bmatrix} r_0(1-u)-d_0 & r_0u \\ r_1v & r_1(1-v)-d_1 \end{bmatrix} = \begin{bmatrix} \lambda_0-r_0u & r_0u \\ r_1v & \lambda_1-r_1v \end{bmatrix}.
\end{align*}
Note that this infinitesimal generator has the same form as the generator (\ref{eq:infgen}) for our baseline model if we define $\mu := r_0u$ and $\nu := r_1v$. This implies that the mean behavior of this new model is identical to the mean behavior of our original model, and all expressions of Appendix \ref{app:generalbranching} apply by making the above substitution. 
The models are distinct, however, on a sample-path basis, which is reflected e.g.~in distinct extinction probabiltities. These probabilities are now determined by the equations   \vspace*{-12pt}

 \begin{align*}
 \begin{split}
& (r_0+d_0) \cdot p_0 = r_0(1-u) p_0^2 + d_0 + r_0u p_0 p_1, \\
& (r_1+d_1) \cdot p_1 = r_1(1-v) p_1^2 + d_1+ r_1v p_0 p_1 ,
\end{split}
\end{align*}
which in general yield different solutions to (\ref{eq:extinctionprob}).

\renewcommand\thefigure{\thesection.\arabic{figure}}
\renewcommand\thetable{\thesection.\arabic{table}}
\renewcommand\theequation{\thesection.\arabic{equation}}
\setcounter{figure}{0} 
\setcounter{equation}{0} 
\setcounter{table}{0} 

\section{
Threshold value for $\mu$
}
\label{app:threshold}

In the main text, we examine the survival probability of a tumor under continuous anti-cancer therapy,
assuming no transiently resistant cell is present at the start of treatment ($m=0$). We observe a threshold value $\mu'$ which is the minimal epimutation rate that guarantees long-term survival of the resistant population. Here, we are interested in deriving an expression for this threshold value.

Let $Y_0(t)$ denote the number of type-0 cells still alive $t$ time units into anti-cancer treatment, and let $Y_1(t)$ denote the number of type-1 cells that have been produced through type-0 epimutations up until time $t$. Since we are only interested in what happens during the initial stages of treatment, and the type-1 population is assumed to be small compared to the initial type-0 population, we ignore reversions from type-1 to type-0 and thus assume that the type-0 population decays at an exponential rate $|\lambda_0|$, with initial population size $Y_0(0)= n$.

Note that given $Y_0(s)$ for $0 \leq s \leq t$, the type-0 population produces a type-1 cell at rate $\mu Y_0(s)$ at time $s$. Therefore, \vspace*{-12pt}

\begin{align*}
E[Y_1(t) | Y_0(s), s \leq t] = \int_0^t \mu Y_0(s)ds,
\end{align*}
from which we conclude \vspace*{-12pt}

\begin{align*}
E[Y_1(t)] &= \int_0^t \mu n e^{\eqrevision{-|\lambda_0|}{\lambda_0}s}ds = \frac{n\mu}{|\lambda_0|} (1-e^{\eqrevision{-|\lambda_0|}{\lambda_0} t}).
\end{align*}
Note that $E[Y_1(t)] \to {n\mu}/|\lambda_0|$ as $t \to \infty$, so approximately \vspace*{-12pt}

\begin{align} \label{eq:epimutationsstart}
n\mu/|\lambda_0|
\end{align}
type-1 cells are created during the initial stages of treatment, during which most of the type-0 population is eradicated. 

Now let $u \ll 1$. The number of type-1 cells, $N_1$, needed to guarantee that the resistant population survives with probability $1-u$ satisfies  \vspace*{-12pt}

\begin{align*}
1-(d_1/r_1)^{N_1} \geq 1-u,
\end{align*}
where we use that $d_1/r_1$ is the extinction probability of a type-1 clone, assuming no reversions to type-0 \cite{durrett2015branching}.
We therefore conclude that \vspace*{-12pt}

\begin{align} \label{eq:N1condition}
N_1 \geq \frac{\log u}{\log(d_1/r_1)}
\end{align}
(recall that $r_1>d_1$ by assumption).
To determine the threshold epimutation rate $\mu'$ that guarantees survival probability of at least $1-u$ under treatment, we combine (\ref{eq:epimutationsstart}) and (\ref{eq:N1condition}) to get \vspace*{-12pt}

\begin{align} \label{eq:threshold1}
\mu' \approx \frac{|\lambda_0| \log u}{n\log(d_1/r_1)}.
\end{align}

\renewcommand\thefigure{\thesection.\arabic{figure}}
\renewcommand\thetable{\thesection.\arabic{table}}
\renewcommand\theequation{\thesection.\arabic{equation}}
\setcounter{figure}{0} 
\setcounter{equation}{0} 
\setcounter{table}{0} 
\section{Time to stable resistance} \label{app:stableresistance}

We are now interested in analyzing the time at which the first successful type-2 cell emerges in the population during continuous anti-cancer therapy.
Let $Z_0(s)$ and $Z_1(s)$ denote the number of type-0 and type-1 cells alive at time $s$, now assuming the presence of permanent resistance mechanisms ($\eta>0$ and $\xi_0,\xi_1>0$). Assume the initial conditions $(Z_0(0),Z_1(0)) = (n,m)$
and that no permanently resistant (type-2) cell is present at the onset.
Let $\eta^{\rm eff}$ be the `effective' reprogramming rate, i.e.~the rate at which epigenetic reprogramming produces a successful type-2 cell (a type-2 cell that gives rise to a clone that does not go extinct), and let $\xi^{\rm eff}_0$ and $\xi^{\rm eff}_1$ be the `effective' mutation rates. Then $\eta^{\rm eff} = \eta\lambda_2/r_2$, $\xi^{\rm eff}_0 = \xi_0\lambda_2/r_2$ and $\xi^{\rm eff}_1 = \xi_1\lambda_2/r_2$, since type-2 cells form a single-type binary branching process with extinction probability $d_2/r_2 = 1-\lambda_2/r_2$
\cite{durrett2015branching}. 

Let $\tau$ denote the time of first occurrence of a successful type-2 cell.
Note that at any time $s$, the total rate at which cells acquire a resistance-conferring mutation is
$r_0\xi_0^{\rm eff} Z_0(s)+r_1\xi_1^{\rm eff}Z_1(s)$,
and the total rate of epigenetic reprogramming is $\eta^{\rm eff} Z_1(s)$. If we condition on the history of the type-0 and type-1 population up until some time $t$, the number of type-2 cells that have emerged by time $t$ is then Poisson distributed with mean  \vspace*{-12pt}

\begin{align*}
\int_0^t (\xi^{\rm eff}_0 r_0 Z_0(s) + \xi^{\rm eff}_1 r_1 Z_1(s) + \eta^{\rm eff} Z_1(s)) ds,
\end{align*}
which implies that \vspace*{-12pt}

\begin{align*}
&P(\tau>t \,|\, (Z_0(s),Z_1(s))_{s \leq t}) \\
=\;& \exp\left(-\int_0^t (\xi^{\rm eff}_0 r_0 Z_0(s) + \xi^{\rm eff}_1 r_1 Z_1(s) + \eta^{\rm eff} Z_1(s)) ds\right),
\end{align*}
since $\{\tau>t\}$ is the event that no type-2 cell has emerged by time $t$.
Taking expectations, we arrive at \vspace*{-12pt}

\begin{align*}
    & P(\tau>t) \\
    =\;& E \left[\exp\left(-\int_0^t (\xi^{\rm eff}_0 r_0 Z_0(s) + \xi^{\rm eff}_1 r_1 Z_1(s) + \eta^{\rm eff} Z_1(s)) ds\right)\right].
\end{align*}
We now argue that $Z_0(s)$ and $Z_1(s)$ in the exponent can be well-approximated by their means. Assuming the initial condition $(Z_0(0),Z_1(0))=(n,m)$, we can write \vspace*{-12pt}

\begin{align*}
(Z_0(s),Z_1(s))=\sum_{j=1}^n\left(Z_{0,j}^{(0)}(s),Z_{1,j}^{(0)}(s)\right) +\sum_{j=1}^m\left(Z_{0,j}^{(1)}(s),Z_{1,j}^{(1)}(s)\right),
\end{align*}
where $Z_{i,j}^{(k)}(t)$ represents the number of type-$i$ cells present at time $t$ descended from a single type-$k$ cell and $\{Z_{i,j}^{(k)}(t): j\in \{1,2,\ldots\}\}$ is an i.i.d.~sequence of such random variables.
Given sufficiently large $n$ and $m$, we know by the law of large numbers that the approximation $(Z_0(s),Z_1(s)) \approx (\phi_0(s),\phi_1(s))$ should be justified. To examine fluctuations around the mean, we note that by the central limit theorem, \vspace*{-12pt}

\begin{align*}
\sum_{j=1}^n\left(Z_{0,j}^{(0)}(s),Z_{1,j}^{(0)}(s)\right)&\approx n\left(E[Z_{0,1}^{(0)}(s)],E[Z_{1,1}^{(0)}(s)]\right)\\
&+\sqrt{n}\left(W_1\sqrt{\Var\big(Z_{0,1}^{(0)}(s)\big)},W_2\sqrt{\Var\big(Z_{1,1}^{(0)}(s)\big)}\right),
\end{align*}
where $(W_1,W_2)$ is a mean-zero bivariate Gaussian. By equation (21) of Section V.7 of \cite{athreya2004branching} we know that there exists a $C>0$ such that \vspace*{-12pt}

\begin{align*}
\sqrt{\Var\big(Z_{i,1}^{(0)}(s)\big)}\leq CE[Z_{i,1}^{(0)}(s)], \quad s >0.
\end{align*}
We can therefore write \vspace*{-12pt}

\begin{align*}
\sum_{j=1}^n\left(Z_{0,j}^{(0)}(s),Z_{1,j}^{(0)}(s)\right)&\approx n\left(E[Z_{0,1}^{(0)}(s)],E[Z_{1,1}^{(0)}(s)]\right) (1+C/\sqrt{n}).
\end{align*}
A similar result holds for the descendants of the type-1 cells. Ignoring the second-order term, we arrive at \vspace*{-12pt}

\begin{align} \label{eq:timetostableresistanceapp}
    & P(\tau>t) \nonumber \\
    \approx\;& \exp\left(-\int_0^t (\xi^{\rm eff}_0 r_0 \phi_0(s) + \xi^{\rm eff}_1 r_1 \phi_1(s) + \eta^{\rm eff} \phi_1(s)) ds\right).
\end{align}
To test the quality of this approximation, we show in Figure \ref{fig:qualityofapproximation} a comparison between \eqref{eq:timetostableresistanceapp} and simulation results for as few as $n=10^4$ starting cells. To calculate the mean number of type-0 and type-1 cells alive at time $s$, $\phi_0(s)$ and $\phi_1(s)$ in \eqref{eq:timetostableresistanceapp}, we can apply  (\ref{eq:meanindividual}) in the main text. Technically,
the birth rates $r_0$ and $r_1$ should be replaced by $r_0(1-\xi_0^{\rm eff})$ and $r_1(1-\xi_1^{\rm eff})$, and the death rate $d_1$ should be replaced by $d_1+\eta^{\rm eff}$, to reflect the introduction of permanent resistance mechanisms,
but assuming $\xi_0,\xi_1 \ll 1$ and $\eta \ll d_1$, we can safely ignore this minor complication.
Using   (\ref{eq:meanindividual}), we can therefore calculate \vspace*{-12pt}

\begin{align} \label{eq:timetostableresistanceapp1}
\begin{split}
\int_0^t \phi_0(s) ds &= -\frac\beta\sigma (1-e^{\sigma t}) - \frac\alpha\rho (1-e^{\rho t}), \\
\int_0^t \phi_1(s) ds &= -\frac{\gamma \beta}\sigma (1-e^{\sigma t}) + \frac{\delta\alpha}\rho (1-e^{\rho t}),
\end{split}
\end{align}
which gives an explicit expression for  (\ref{eq:timetostableresistanceapp}) assuming $\sigma \neq 0$ and $\rho \neq 0$. Note that the expressions in (\ref{eq:timetostableresistanceapp1}) can be interpreted as the `total mass' of type-0 and type-1 cells, respectively, up until time $t$.

\begin{figure}[t!]
\includegraphics[scale=1]{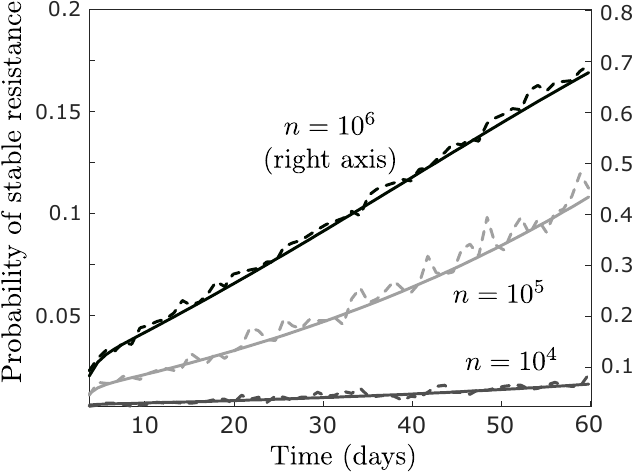}
\centering
\caption{
Comparison of expression \eqref{eq:timetostableresistanceapp} (whole curves) with simulation results (dotted curves) in the baseline parameter regime, assuming a starting population of $n=10^4$ (bottom curve; left axis), $n=10^5$ (middle curve; left axis) or $n=10^6$ cells (top curve; right axis). Simulation results are based on 2000 runs of the process. Other parameter values are $r_0 = 0.04$ (per hour), $d_0 = 0.08$, $r_1 = 0.0162$, $d_1 = 0.015$, $\mu = 4 \cdot 10^{-5}$, $\nu = 4 \cdot 10^{-4}$, $\eta = 4 \cdot 10^{-7}$, $\xi_0 = \xi_1 = 10^{-7}$ (per cell division), and $m=0$.
}
\label{fig:qualityofapproximation}
\end{figure}

In the main text, we apply (\ref{eq:timetostableresistanceapp}) and (\ref{eq:timetostableresistanceapp1}) to evaluate interval treatment strategies, where the anti-cancer agent is intermittently applied and removed. This can create a problem, since our calculations of $\eta^{\rm eff}$, $\xi_0^{\rm eff}$ and $\xi_1^{\rm eff}$ assume that the rates $r_2$ and $d_2$ of birth and death of type-2 cells do not change over time. However, since we make the assumption in our baseline parameter regime that $r_2$ and $d_2$ are identical on and off the anti-cancer agent, we do not have to make any adjustments to  (\ref{eq:timetostableresistanceapp}) to use it to evaluate interval treatment strategies. In fact, since we assume $r_2=0.04$ and $d_2=0.0015$ with $\lambda_2/d_2 \approx 1$, we can effectively take $\eta^{\rm eff} = \eta$, $\xi_0^{\rm eff} = \xi_0$ and $\xi_1^{\rm eff} = \xi_1$.

To determine whether stable resistance develops at all during continuous anti-cancer therapy, we consider the event $\{\tau=\infty\}$. It is possible to estimate its probability by taking $t \to \infty$ in (\ref{eq:timetostableresistanceapp}), but then we first have to condition on whether the overall population of type-0 and type-1 cells goes extinct or not, since $\phi_0(s)$ and $\phi_1(s)$ blow up as $s \to \infty$ for any supercritical process, irrespective of the likelihood that the population survives.
However, it is more straightforward to proceed as follows. Let $p_0$, $p_1$ and $p_2$ denote the extinction probabilities of clones started by a single cell of each type. By conditioning on the initial event, we derive the following conditions for these probabilities: \vspace*{-12pt}

 \begin{align}  \label{eq:extinctionprobperm}
 \begin{split}
& (r_0+d_0+\mu) \cdot p_0 = r_0(1-\eqrevision{\xi_0}{\xi}) p_0^2 + d_0 + \mu p_1 + r_0\eqrevision{\xi_0}{\xi} p_0 p_2, \\
& (r_1+d_1+\nu+\eta) \cdot p_1 = r_1 (1-\eqrevision{\xi_1}{\xi}) p_1^2 + d_1 + \nu p_0 + r_1 \eqrevision{\xi_1}{\xi} p_1 p_2 + \eta p_2,\\
& p_2 = d_2/r_2,
\end{split}
\end{align}
where we use that fact that type-2 cells form a single-type binary branching process with extinction probability $d_2/r_2$. Note that these equations reduce to (\ref{eq:extinctionprob})
if we set $\xi_0=\xi_1=0$ and $\eta=0$.
We can then calculate
\begin{align}
P(\tau=\infty) = p_0^n p_1^m,
\end{align}
which gives the probability that stable resistance does not develop at all during continuous administration of anti-cancer treatment. The probability that stable resistance does develop is then
\begin{align} \label{eq:survivalprobperm}
P(\tau<\infty) = 1-p_0^n p_1^m.
\end{align}

\renewcommand\thefigure{\thesection.\arabic{figure}}
\renewcommand\thetable{\thesection.\arabic{table}}
\renewcommand\theequation{\thesection.\arabic{equation}}
\setcounter{figure}{0} 
\setcounter{equation}{0} 
\setcounter{table}{0} 
\section{Robustness analysis} \label{app:robustness}

\subsection*{Survival probability}

Table \ref{table:sensitivity1} shows how survival probability, both at the indidividual cell level and for the tumor population as a whole, responds to changes in $r_0$, $d_0$, $r_1$ and $d_1$, pertinent to the discussion of Section \ref{sec:solelyswitching}.
\begin{table}[t]
    \centering
    \begin{tabular}{| m{0.55cm}  m{0.55cm}  m{0.7cm}  m{1cm} | m{1.05cm}  m{1.05cm}  m{1.5cm} |}
    \hline
         $r_0$&$d_0$&$r_1$&$d_1$&$1-p_0$&$1-p_1$&$1-p_0^np_1^m$  \\
         \hline
         0.04&0.08&0.0162&0.015&0.0049\%&4.94\%&100\% \\
         \hline
         +10\%&&&&0.0055\%&4.94\%&100\% \\
&+10\%&&&0.0041\%&4.94\%&100\% \\
+10\%&+10\%&&&0.0045\%&4.94\%&100\% \\
\hline
&&$+10\%$&&0.0136\%&13.58\%&100\% \\
&&&$+10\%$&0\%&0\%&0\% \\
&&$+10\%$&$+10\%$&0.0052\%&5.16\%&100\% \\
\hline
+50\%&&&&0.0099\%&4.94\%&100\% \\
&+50\%&&&0.0025\%&4.94\%&100\% \\
+50\%&+50\%&&&0.0033\%&4.94\%&100\% \\
\hline
&&$+50\%$&&0.0366\%&36.63\%&100\% \\
&&&$+50\%$&0\%&0\%&0\% \\
&&$+50\%$&$+50\%$&0.0058\%&5.76\%&100\% \\
\hline
    \end{tabular}
    \caption{Response of the survival probability of (i) a clone started by a single type-0 cell ($1-p_0$), (ii) a clone started by a single type-1 cell ($1-p_1$), and (iii) the overall population of $n=10^6$ type-0 and $m=0$ type-1 cells ($1-p_0^np_1^m$), to changes in $r_0$, $r_1$, $d_0$ and $d_1$. The probabilities $p_0$ and $p_1$ can be calculated using \eqref{eq:extinctionprob}. In the top row, we show the baseline parameter values for $r_0$, $d_0$, $r_1$ and $d_1$, and the survival probabilities given these values. In the remaining lines, we adjust one or two parameters at a time as indicated and display survival probabilities given these changes. All rates are measured per hour.}
    \label{table:sensitivity1}
\end{table}

\subsection*{Threshold value for $\mu$}

Table \ref{table:sensitivity2} shows how the threshold value for $\mu$, denoted as $\mu'$ and given by \eqref{eq:threshold}, responds to to changes in $r_1$ and $d_1$, pertinent to the discussion of Section \ref{sec:tumorsurvival}.
\begin{table}[t]
    \centering
    \begin{tabular}{| m{0.8cm}  m{1cm} | m{2cm}  |}
    \hline
         $r_1$&$d_1$&$\mu'$  \\
         \hline
        0.0162&0.015&$3.59 \cdot 10^{-6}$  \\
         \hline
$+10\%$&&$1.60 \cdot 10^{-6}$ \\
&$+10\%$&$+\infty$ \\
$+10\%$&$+10\%$&$3.59 \cdot 10^{-6}$ \\
\hline
$+50\%$&&$0.57 \cdot 10^{-6}$ \\
&$+50\%$&$+\infty$ \\
$+50\%$&$+50\%$&$3.59 \cdot 10^{-6}$ \\
\hline
    \end{tabular}
    \caption{Response of the threshold value $\mu'$, as given by \eqref{eq:threshold}, to changes in $r_1$ and $d_1$. We use $+\infty$ to indicate that there is no value for $\mu$ for which tumor survival is guaranteed, which is the case when the population is subcritical. All rates are measured per hour.}
    \label{table:sensitivity2}
\end{table}

\subsection*{Evaluation of combination treatment strategies}

We now discuss various possible changes to the assumptions underlying Figure \ref{fig:interval} in the main text. In this section, we let $\mu_d$ and $\nu_d$ denote the rates of epimutation and reversion when the epigenetic drug is being applied to distinguish them from the rates $\mu$ and $\nu$ off the epigenetic drug.

\subsubsection*{Tumor size and rates of mutation and reprogramming}

In the main text, we assume an initial population size of $n = 10^6 \cdot 0.999$ and $m = 10^6 \cdot 0.001$ cells, an epigenetic reprogramming rate of $\eta = 4 \cdot 10^{-7}$ per hour, and a mutation rate of $\xi_0=\xi_1=10^{-7}$ per cell division. If we assume that the rate of epigenetic reprogramming is increased by an order of magnitude to $\eta = 4 \cdot 10^{-6}$, the relative attractiveness of Schedule B over the other schedules becomes even more pronounced than in the baseline case (Fig \ref{fig:Treatment_appendix1}a).
If we also increase the mutation rate an order of magnitude to $\xi_0=\xi_1=10^{-6}$, it becomes more difficult to eradicate the tumor, but the relative attractiveness of the schedules remains unchanged (Fig \ref{fig:Treatment_appendix1}b).

If we increase the population size one order of magnitude to $n = 10^7 \cdot 0.999$ and $m = 10^7 \cdot 0.001$, assuming the baseline values $\eta = 4 \cdot 10^{-7}$ and $\xi_0=\xi_1=10^{-7}$ for the rates of reprogramming and mutation (figure not shown), the effects will be similar to those shown in Figure \ref{fig:Treatment_appendix1}b, where the population size is as in the baseline case but both rate parameters are increased an order of magnitude.
Thus, the relationship between initial population size and the rate of mutation and epigenetic reprogramming is more important than the absolute value of each in isolation, and any effect of increasing the population size can be negated by a proportional decrease in the these rates.
We also note that if the anti-cancer agent being applied reduces the birth rate $r_0$ of type-0 cells, this will have the same effect as reducing the mutation rate $\xi_0$
(see discussion in Appendix \ref{app:parametrization} above).

For an even larger initial tumor size of $n = 10^8 \cdot 0.999$ and $m = 10^8 \cdot 0.001$, a resistance-conferring mutation will be guaranteed to arise no matter which schedule is applied (figure not shown), although we recover the same behavior as for the $n= 10^7 \cdot 0.999$ and $m = 10^7 \cdot 0.001$ case if the change in size is accompanied by a smaller reprogramming rate, $\eta = 4 \cdot 10^{-8}$, and a smaller mutation rate, $\xi_0=\xi_1=10^{-8}$.
However, even if the initial population size is $n = 10^8 \cdot 0.999$ and $m = 10^8 \cdot 0.001$, with unchanged rates of mutation and reprogramming,
there can still be benefits to applying a combination of an anti-cancer agent and an epigenetic drug, since Schedule B can prevent resistance conferred through epigenetic reprogramming (Fig \ref{fig:Treatment_appendix1}c). In fact, if we assume that the epigenetic drug is stronger than in the main text, so that it increases the reversion rate $\nu$ by three orders of magnitude and decreases the epimutation rate to the same extent, then Schedule B is able to prevent resistance through epigenetic reprogramming with high probability, whereas resistance is effectively guaranteed under any other schedule (Fig \ref{fig:Treatment_appendix1}d).

\begin{figure}
\hspace*{-1cm} \includegraphics[scale=1]{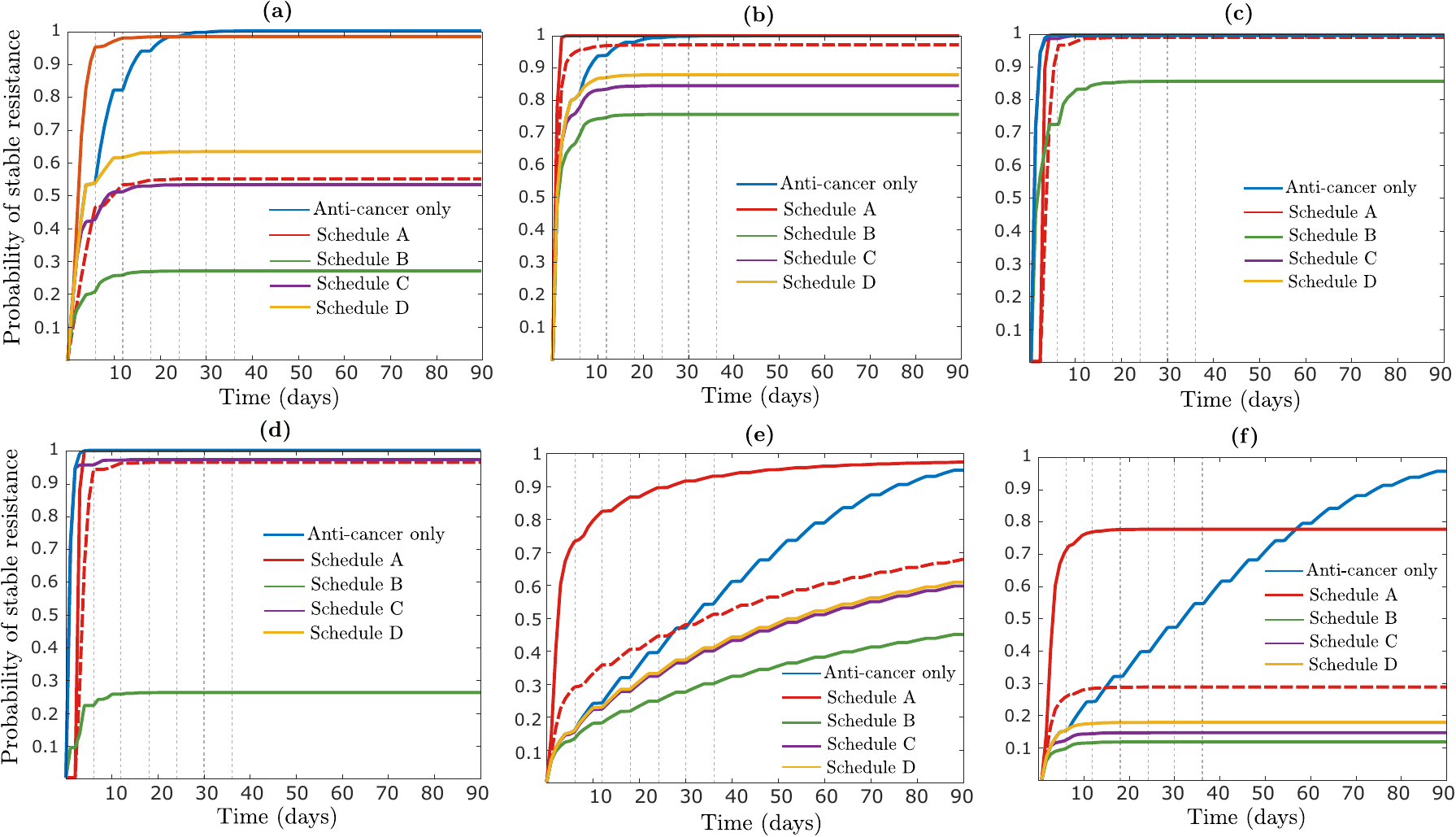}
\centering
\caption{
Time-evolution of the probability of successful type-2 emergence for the treatment schedules examined in Figure \ref{fig:interval} in the main text. Parameters behave as in Figure \ref{fig:interval}
except in (a), $\eta = 4 \cdot 10^{-6}$ (per hour), in (b), $\eta = 4 \cdot 10^{-6}$ and $\xi_0=\xi_1=10^{-6}$ (per cell division), in (c), $n=10^8 \cdot 0.999$, $m = 10^8 \cdot 0.001$ and $\xi_0=\xi_1=0$, in (d), $n=10^8 \cdot 0.999$, $m = 10^8 \cdot 0.001$, $\xi_0=\xi_1=0$, $\mu_d = \mu \cdot 10^{-3}$ and $\nu_d = \nu \cdot 10^3$, in (e), $\mu_d = \mu \cdot 10^{-1}$ and $\nu_d = \nu \cdot 10^1$, and in (f), $\mu_d = \mu \cdot 10^{-3}$ and $\nu_d = \nu \cdot 10^3$.
}
\label{fig:Treatment_appendix1}
\end{figure}

\subsubsection*{Different epigenetic drug action}

In the main text, we assume that the epigenetic drug increases the reversion rate $\nu$ by two orders of magnitude, and decreases the epimutation rate $\mu$ to the same extent.
In Figures \ref{fig:Treatment_appendix1}e and \ref{fig:Treatment_appendix1}f, we show results first assuming that the change in each is one order of magnitude (Fig \ref{fig:Treatment_appendix1}e), and then three orders of magnitude (Fig \ref{fig:Treatment_appendix1}f).
We note that Schedule B continues to be the most attractive schedule, and in Figure \ref{fig:Treatment_appendix1}e, the difference between Schedule B and Schedules C,D becomes more pronounced than for the case shown in the main text.

\subsubsection*{Faster or slower switching dynamics}

In the main text, we assume an epimutation rate of $\mu = 4 \cdot 10^{-5}$ in the absence of the epigenetic drug, in accordance with our baseline parameter regime. In Figures \ref{fig:Treatment_appendix2}a and \ref{fig:Treatment_appendix2}b, we show results assuming $\mu = 4 \cdot 10^{-4}$ (Fig \ref{fig:Treatment_appendix2}a) and $\mu = 4 \cdot 10^{-6}$ (Fig \ref{fig:Treatment_appendix2}b), with the epimutation rate $\mu_d$ in the presence of the epigenetic drug obtained by reducing $\mu$ by two orders of magnitude in each case. As before, the qualitative dynamics remain similar, with the difference between Schedule B and the other schedules even more pronounced in Figure \ref{fig:Treatment_appendix2}a than for the baseline case.

\subsubsection*{Different treatment block sizes}

In the main text, we assume that each treatment block is 2 days.
In Figures \ref{fig:Treatment_appendix2}c and \ref{fig:Treatment_appendix2}d, we show results assuming that each block is 1 day (Fig \ref{fig:Treatment_appendix2}c) or 3 days (Fig \ref{fig:Treatment_appendix2}d), and observe similar results as before.

\subsubsection*{Resistant phenotype is not at a selective disadvantage off drug}
In our baseline regime, we assume that the resistant type-1 phenotype is at a selective disadvantage to drug-sensitive type-0 cells in the absence of the anti-cancer agent. Figure \ref{fig:Treatment_appendix2}e reveals that our results do not depend on this assumption, as if we assume $r_1 = r_0 = 0.04$ and $d_1 = d_0 = 0.0015$ off the anti-cancer agent, Schedule B continues to perform the best.

\begin{figure}
\hspace*{-1cm} \includegraphics[scale=1]{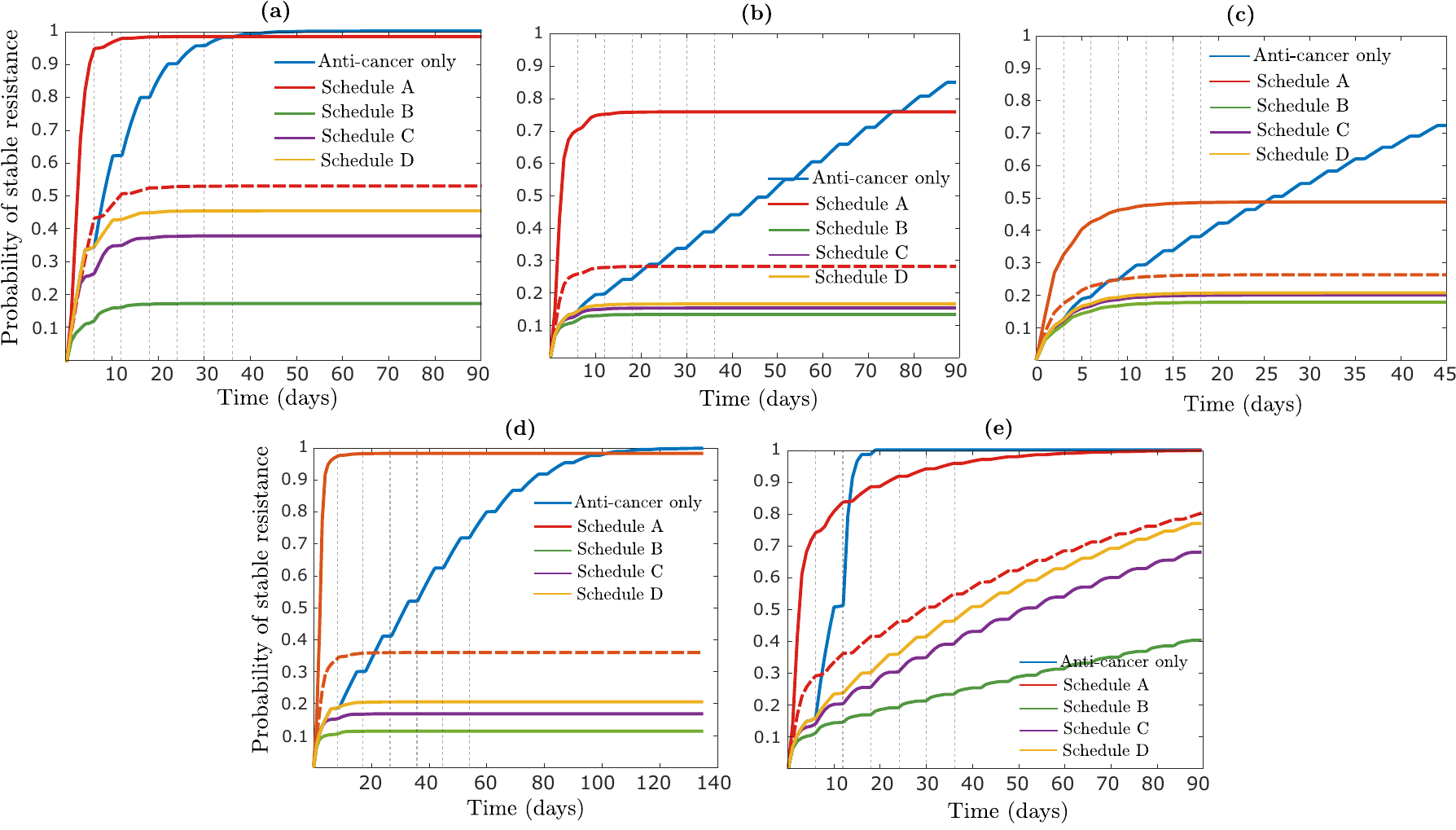}
\centering
\caption{
Time-evolution of the probability of successful type-2 emergence for the treatment schedules examined in Figure \ref{fig:interval} in the main text. Parameters behave as in Figure \ref{fig:interval}
except in (a), $\mu = 4 \cdot 10^{-4}$ in the absence of epigenetic drug, in (b), $\mu = 4 \cdot 10^{-6}$ in the absence of epigenetic drug, in (c), the size of each treatment block is 1 day as opposed to 2 days, in (d), the size of each block is 3 days, and in (e), $r_1 = 0.04$ (per hour) and $d_1 = 0.0015$ in the absence of the anti-cancer agent.
}
\label{fig:Treatment_appendix2}
\end{figure}

\end{appendices}

\end{document}